\documentclass[12pt,draftclsnofoot,onecolumn]{IEEEtran}
\usepackage{yfonts,color}
\usepackage{soul,xcolor}
\usepackage{bm}
\include{graphics}
\include{amsmath}
\usepackage[normalem]{ulem}
\usepackage{multirow,enumitem}
\usepackage{algorithm}
\usepackage{algpseudocode}
\usepackage{verbatim}
\usepackage[latin1]{inputenc}
\usepackage{amsmath}
\usepackage{amsfonts}
\usepackage{amssymb}
\usepackage{mathrsfs}
\usepackage{booktabs}
\usepackage{url}
\usepackage[skip=1pt,font=scriptsize]{subcaption}
\usepackage{ifthen}
\usepackage{xspace}
\usepackage{dsfont}
\usepackage{bbm}
\usepackage{cite}
\usepackage{epsfig}
\usepackage{array}
\usepackage{multicol}
\usepackage{algpseudocode}
\usepackage{algorithm}
\usepackage{amssymb}
\usepackage{breqn}
\usepackage{esint}
\usepackage[skip=2pt,font=footnotesize]{caption}
\usepackage{multicol}
\usepackage{amsthm}
\usepackage{setspace}
\usepackage{lipsum}
\usepackage{dblfloatfix}
\theoremstyle{plain}
\newtheorem{theorem}{Theorem}
\newtheorem{lemma}{Lemma}
\theoremstyle{definition}
\newtheorem{defi}{Definition}
\newtheorem{corollary}{Corollary}

\theoremstyle{remark}

\usepackage{footnote}
\newcommand\numberthis{\addtocounter{equation}{1}\tag{\theequation}}
\DeclareMathOperator*{\argmin}{argmin}

\newcommand{\nm}[1]{\textcolor{blue}{\textbf{NM-comment: [#1]}}}
\newcommand{\sst}[1]{\st{#1}}

\newcommand{\nma}[1]{{\color{blue}#1}}

\newcommand{\size}[1]{| #1 |}
\newcommand{\supp}{\mathrm{supp}}
\newcommand{\stkout}[1]{\ifmmode\text{\color{red}\sout{\ensuremath{#1}}}\else\sout{#1}\fi}
\begin{document}
\setstcolor{red}
\setulcolor{red}
\setul{red}{2pt}
\title{Energy-Efficient  Interactive Beam-Alignment\\for Millimeter-Wave Networks}
\author{Muddassar Hussain, and Nicolo Michelusi
\thanks{M. Hussain and N. Michelusi are with the School of Electrical and Computer Engineering, Purdue University. \emph{email}: \{hussai13, michelus\}@purdue.edu. This research has been funded by NSF under grant CNS-1642982.
 Part of this work appeared at Asilomar'18 \cite{asilomar2018}.}
\vspace{-12mm}
}
\maketitle
\begin{abstract}
Millimeter-wave will be a key technology in next-generation wireless networks thanks to abundant bandwidth availability. However, the use of large antenna arrays with beamforming demands precise beam-alignment between transmitter and receiver, and may entail huge overhead in mobile environments. This paper investigates the design of an  optimal interactive beam-alignment and data communication protocol, with the goal of minimizing power consumption under a minimum rate constraint. The base-station selects beam-alignment or data communication and the beam parameters, based on feedback from the user-end. Based on the sectored antenna model and uniform prior on  the angles of departure and arrival (AoD/AoA), the optimality of a \emph{fixed-length} beam-alignment phase followed by a data-communication phase is demonstrated. Moreover, a \emph{decoupled fractional} beam-alignment method is shown to be optimal, which decouples over time the alignment of AoD and AoA, and iteratively scans a fraction of their region of uncertainty. A heuristic policy is proposed for non-uniform prior on AoD/AoA, with provable performance guarantees, and it is shown that the uniform prior is the worst-case scenario. The performance degradation due to detection errors is studied analytically and via simulation. The numerical results with analog beams depict {up to $4 \mathrm{dB}$, $7.5 \mathrm{dB}$, and $14 \mathrm{dB}$  gains over a state-of-the-art bisection method}, conventional and interactive exhaustive search policies, respectively,
 {and demonstrate that the sectored model provides valuable insights for beam-alignment design.}
\end{abstract}
\vspace{-5mm}
\section{Introduction}
Mobile traffic has witnessed a tremendous growth over the last decade, 18-folds over the past five years alone, and is expected to grow
 with a compound annual growth rate of 47\% from 2016 to 2021 \cite{cisco2}.
This rapid increase poses  a severe burden to current systems operating below 6 GHz, due to limited bandwidth availability.
Millimeter-wave (mm-wave) is emerging as a promising solution to enable multi-Gbps communication, thanks to abundant bandwidth availability \cite{channel_model}. 
However, high isotropic path loss and sensitivity to blockages
pose challenges in
 supporting high capacity and mobility~\cite{rappaport_mmwave_book}.
To overcome the path loss, mm-wave systems will thus leverage narrow-beams, by using large antenna arrays at both base stations (BSs) and user-ends (UEs).

\par Nonetheless, narrow transmission and reception beams are susceptible to frequent loss
of alignment, due to mobility or blockage, which necessitate the use of beam-alignment protocols. Maintaining beam-alignment between transmitter and receiver can be
challenging, especially in mobile scenarios, and may entail significant overhead, thus potentially offsetting the benefits of mm-wave directionality. Therefore, it is imperative to 
design schemes to mitigate its overhead. 
\par To address this challenge, in our previous work \cite{icc2018,ita2017, ita2018,asilomar2017}, we address the optimal design of beam-alignment protocols.
In~\cite{icc2018}, we optimize the trade-off between data communication and beam-sweeping, under the assumption of an exhaustive search method, in a mobile scenario where the BS widens its beam to mitigate the uncertainty on the UE position. 
 In \cite{ita2017,ita2018}, we design a throughput-optimal beam-alignment scheme for one and two UEs, respectively, and we prove the optimality of a \emph{bisection search}.   However, the model therein does not consider the energy cost of beam-alignment, which may be significant when targeting high detection accuracy. {It is noteworthy that, if the energy consumption of beam-alignment is small, bisection search is the best policy since it is the fastest way to reduce the uncertainty region of the angles of arrival (AoA) and departure (AoD). For this reason, it has been employed in previous works related to multi-resolution codebook design, such as \cite{new_benchmark}}.  In \cite{asilomar2017,asilomar2018}, we incorporate the energy cost of beam-alignment, and prove the optimality of a \emph{fractional search} method.
 Yet, in \cite{ita2017, asilomar2017, ita2018, icc2018}, optimal design is carried out under restrictive assumptions that the UE receives isotropically, and that the duration of beam-alignment is fixed. In practice, the BS may switch to data transmission upon finding a strong beam, as in \cite{interactiveexhaustive}, and \emph{both} BS and UE
 may use narrow beams to fully leverage the beamforming gain.
 
To the best of our knowledge, the optimization of \emph{interactive}  beam-alignment,
\emph{jointly} at both BS and UE, is still an open problem.
Therefore, in this paper, we consider a more flexible model than our previous papers \cite{ita2017, asilomar2017, ita2018, icc2018}, by allowing dynamic switching between beam-alignment and data-communication and joint optimization over BS-UE beams, BS transmission power and rate. 
\emph{Indeed, we prove that a fixed-length beam-alignment scheme followed by data communication is optimal, 
and we prove the optimality of a decoupled fractional search method, which decouples over time the alignment of AoD and AoA, and iteratively scans a fraction of their region of uncertainty.}
Using Monte-Carlo simulation with analog beams, we demonstrate superior performance, with
 {up to $4 \mathrm{dB}$, $7.5 \mathrm{dB}$, and $14 \mathrm{dB}$ power gains over the state-of-the-art bisection method~\cite{new_benchmark}}, conventional exhaustive, and interactive exhaustive search policies, respectively.
Compared to our recent paper \cite{asilomar2018}, the system model adopted in this paper is more realistic since it captures the effects of fading and resulting outages,\
non-uniform priors on AoD/AoA, and detection errors.
Additionally, the model in \cite{asilomar2018} is restricted to a two-phase protocol with deterministic beam-alignment duration. In this paper, we show that this is indeed optimal.
\vspace{-2mm}
 \subsection{Related Work}
  Beam-alignment has been a subject of intense research due to its importance in mm-wave communications. The research in this area can be categorized into beam-sweeping \cite{exhaustive,iterative,ita2017,asilomar2017, ita2018, icc2018,new_benchmark,MDP_ba}, data-assisted schemes \cite{radar,lowfreq,va,inverse_finger},
and AoD/AoA estimation \cite{alkhateeb,marzi}. The simplest and yet most popular beam-sweeping scheme is \emph{exhaustive} search  \cite{exhaustive}, which sequentially scans through all possible BS-UE beam pairs and selects the one with maximum signal power. A version of this scheme has been adopted in existing mm-wave standards including IEEE 802.15.3c  \cite{ieee80215c} and IEEE 802.11ad \cite{ieee80211ad}.
An interactive version of exhaustive search has been proposed in \cite{interactiveexhaustive}, wherein the beam-alignment phase is terminated once the power of  the received beacon is above a certain threshold. 
 The second popular scheme is \emph{iterative search}  \cite{iterative}, where scanning is first performed using wider beams followed by refinement using narrow beams. 
A variant of \emph{iterative search} is studied in \cite{multi_res}, where
the beam sequence is chosen adaptively from a pre-designed multi-resolution codebook. However, this codebook is designed independently of the beam-alignment protocol, thereby potentially resulting in suboptimal design. 
In \cite{MDP_ba}, the authors consider the design of a beamforming vector sequence
based on a partially observable (PO-) Markov decision processes (MDPs).  However, 
POMDPs are generally not amenable to closed-form solutions, and have high complexity.
To reduce the computational overhead, the authors focus on a greedy algorithm, which yields a sub-optimal policy. 

\par Data-aided schemes utilize the information from sensors to aid beam-alignment and reduce the beam-sweeping cost (e.g., from radar \cite{radar},
 lower frequencies \cite{lowfreq}, position information \cite{inverse_finger,va}).
AoD/AoA estimation schemes leverage the sparsity of mm-wave channels, and include compressive sensing schemes \cite{alkhateeb} or approximate maximum likelihood estimators~\cite{marzi}.
\label{page6}
 {In \cite{bf_tradeoff}, the authors compare different schemes and conclude that the performance of beam-sweeping is comparable with the best performing estimation schemes based on compressed sensing. Yet, beam-sweeping has the added advantage of low complexity over compressed sensing schemes, which often involve solving complex optimization problems, and is more amenable to analytical insights on the beam-alignment process. For these reasons, in this paper we focus on beam-sweeping, and derive insights on its optimal design.}

 All of the aforementioned schemes choose the beam-alignment beams from pre-designed codebooks,
 use heuristic protocols, {or are not amenable to analytical insights.} By choosing the beams from  a restricted beam-space or a predetermined protocol, 
 optimality may not be achieved. Moreover, all of these papers do not consider the energy and/or time overhead of beam-alignment as part of their design.
 In this paper, we address these open challenges by optimizing the
  beam-alignment protocol to maximize the communication performance.
  \vspace{-2mm}
 \subsection{Our Contributions}
 \begin{enumerate}[leftmargin=.60cm]
 \item Based on a MDP formulation, under the \emph{sectored} antenna model \cite{sectored_model}, uniform AoD/AoA prior, and small detection error assumptions, we prove the optimality of a \emph{fixed-length} two-phase protocol, with a beam-alignment phase  of fixed duration followed by a data communication phase. We provide an algorithm to compute the optimal duration. 
 \item  We prove the optimality of  a \emph{decoupled fractional search} method, 
 which scans a fixed fraction of the region of uncertainty of the AoD/AoA in each beam-alignment slot. Moreover, the beam refinements over the AoD and AoA dimensions are decoupled over time, thus proving the sub-optimality of \emph{exhaustive search} methods.
 \item Inspired by the decoupled fractional search method, we propose a heuristic scheme for the case of non-uniform prior on AoD/AoA with provable performance,
 and prove that the  uniform prior is indeed the worst-case scenario. 
 \item We analyze the effect of detection errors on the performance of the proposed protocol.
 \item We evaluate its performance via simulation using analog beams, and demonstrate up to $4 \mathrm{dB}$, $7.5 \mathrm{dB}$, and $14 \mathrm{dB}$ power gains compared to the state-of-the-art bisection scheme \cite{new_benchmark}, conventional and interactive exhaustive search policies, respectively. 
 {Remarkably, the sectored model provides valuable insights for beam-alignment design.}
 \end{enumerate}
 
\par The rest of this paper is organized as follows. In Secs.~\ref{sec:sysmodel}, we describe the system model. In Sec. \ref{sec:probform}, we formulate the optimization problem. In Secs. \ref{sec:unifprior}-\ref{sec:nonunifprior}, we provide the analysis for the case of uniform and non-uniform priors on AoD/AoA.
    In Sec. \ref{sec:fa_md_impact}, we analyze the effects of false-alarm and misdetection errors.
     The numerical results are provided in Sec. \ref{numres}, followed by concluding remarks in Sec. \ref{sec:concl}.
     The main analytical proofs are provided in the Appendix.

\begin{figure}[!t]
\centering
  \includegraphics[width=.6\linewidth,trim={25 15 35 20},clip]{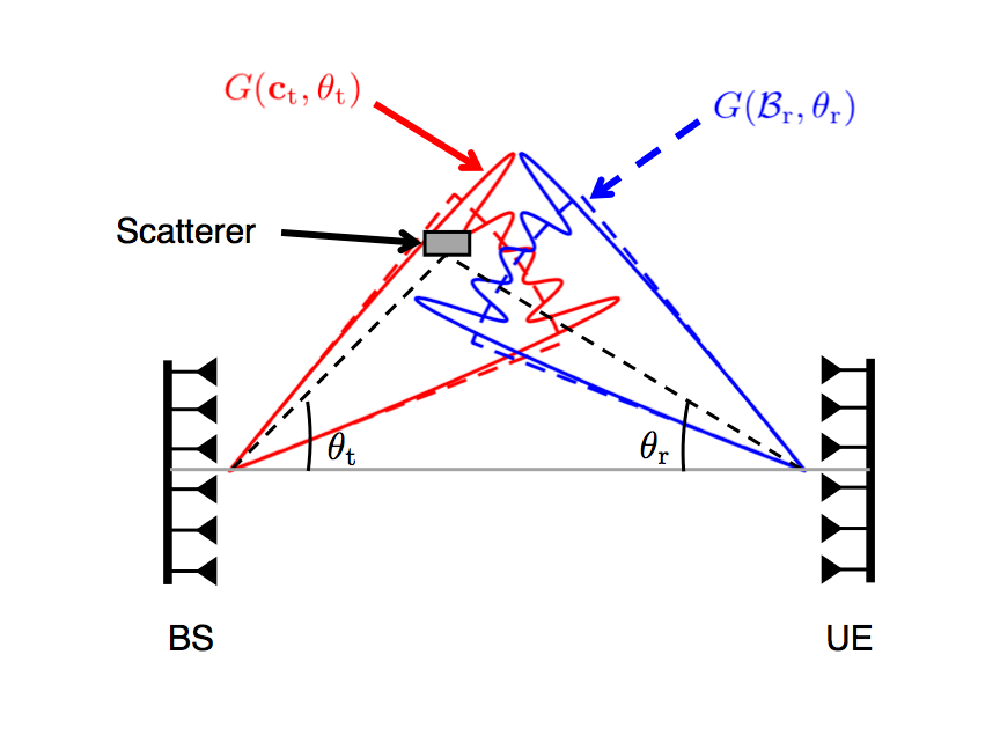}
\caption{{Actual beam pattern $G(\mathbf c_{\mathrm x},\theta_{\mathrm x})$ generated using the algorithm in~\cite{rate_maximization2} with $M_{\mathrm t}{=}M_{\mathrm r}{=}128$  antennas. (solid lines)
versus sectored model $G(\mathcal B_{\mathrm x},\theta_{\mathrm x})$ (dashed lines) \cite{sectored_model}, on a linear scale.
Sidelobes are not visible due to their small magnitude.}}
\label{fig:beam_pattern}
\vspace{-7mm}
\end{figure}
\vspace{-7mm}
\section{System Model }
\label{sec:sysmodel}
We consider a downlink scenario in a mm-wave cellular system with one base-station (BS) and one mobile user (UE) at distance $d$ from the BS, both equipped with uniform linear arrays (ULAs) with $M_{\mathrm t}$ and $M_{\mathrm r}$  antennas, respectively, depicted in Fig.~\ref{fig:beam_pattern}.
Communication occurs over frames of fixed duration $T_{\mathrm{fr}}$, each composed of $N$ slots indexed by $\mathcal I{\equiv}\{0,1, \ldots, N{-}1\}$ of duration $T{=}T_{\mathrm{fr}}/N$, each carrying $S$ symbols of duration $T_{\mathrm{sy}}{=}T/S$. 
Let $s$ be the transmitted symbol, with $\mathbb E[|s|^2]{= }1$. Then, the signal received at the UE is
\begin{align}
\label{eq:receivedsignal}
 y = \sqrt{P}\mathbf c_{\mathrm r}^H \mathbf H \mathbf c_{\mathrm t}  s + \mathbf c_{\mathrm r}^H {\mathbf w},
\end{align}
where $P$ is the average transmit power of the BS; $\mathbf H{\in}\mathbb C^{M_{\mathrm r} \times M_{\mathrm t}}$ is the channel matrix; $\mathbf c_{\mathrm t}{\in}\mathbb C^{M_{\mathrm t}}$ is the BS beam-forming vector; $\mathbf c_{\mathrm r}{\in}\mathbb C^{M_{\mathrm r}}$ is the UE combining vector; ${\mathbf w}\sim \mathcal{CN}(\boldsymbol 0,N_0W_{\mathrm{tot}} \mathbf I)$ is additive white Gaussian noise (AWGN). The symbols $N_0$ and $W_{\mathrm{tot}}$ denote the one-sided power spectral density of AWGN and the system bandwidth, respectively. By assuming analog beam-forming at both BS and UE, $\mathbf c_{\mathrm t}$ and $\mathbf c_{\mathrm r}$ satisfy the unit norm constraints $\| \mathbf c_{\mathrm t}\|_2^2 =\|\mathbf c_{\mathrm r}\|_2^2=1$.  
{
\label{p1}
The channel matrix $\mathbf H$ follows the
 extended Saleh-Valenzuela geometric model \cite{extended},
\begin{align}
\label{eq:channel_matrixx}
\mathbf H = \sqrt{\frac{M_{\mathrm t} M_{\mathrm r}}{K} }\sum_{\ell=1}^{K} h_{\ell}\;  \mathbf d_{\mathrm r}(\theta_{\mathrm r,\ell})\mathbf d_{\mathrm t}^H(\theta_{\mathrm t,\ell}),
\end{align}
where $h_{\ell} \in \mathbb C$, $\theta_{\mathrm t,\ell}$ and $\theta_{\mathrm r,\ell}$ denote the small scale fading coefficient, AoD and AoA of the $\ell^{th}$ cluster, respectively.}
The terms $\mathbf d_{\mathrm r}(\theta_{\mathrm r,\ell}){\in}\mathbb C^{M_{\mathrm r}}$
{and} $\mathbf d_{\mathrm t}(\theta_{\mathrm t,\ell}){\in}\mathbb C^{M_{\mathrm t}}$
 are the UE and BS array response vectors, respectively. For ULAs,
 {$\theta_{\mathrm t,\ell}$ (respectively, $\theta_{\mathrm r,\ell}$) is the angle formed between the outgoing (incoming) rays of the $l$th channel cluster and the perpendicular to the BS (UE) antenna array,  as represented in Fig.~\ref{fig:beam_pattern}, so that
\begin{align}
\nonumber
\mathbf  d_{\mathrm x }(\theta_{\mathrm x }) = &\frac{1}{\sqrt{M_{{\mathrm x}}}}\left[1, e^{j\frac{2\pi d_{\mathrm x}}{\lambda} \sin \theta_{\mathrm x }} ,\ldots,  e^{j(M_{{\mathrm x}}-1)\frac{2\pi d_{\mathrm x}}{\lambda} \sin \theta_{\mathrm x }} \right]^{\top},
\end{align}
 where $\mathrm x\in \{{\mathrm t },{\mathrm r}\}$, $d_{\mathrm t}$ and $d_{\mathrm r}$ are the antenna spacing of the BS and UE arrays, respectively, $\lambda$ is {the} wavelength of {the} carrier signal. 
{In \eqref{eq:channel_matrixx}, $K\geq\mathrm{rank}(\mathbf H)$ is the total number of clusters.
Note that $\mathbf H$ has low-rank if $K\ll \min\{M_{{\mathrm t}},M_{{\mathrm r}}\}$.
 In this paper, we assume that there is a single dominant cluster ($K{=}1$).  This assumption has been adopted in several previous works (e.g., see \cite{zorzi,jeffery_andrew_ba}), and is motivated by channel measurements and modeling works such as \cite{rappaport_channel_model}, where it is shown that, in \emph{dense urban environments}, 
with high probability the mm-wave channel exhibits only one or two clusters, with the dominant one containing most of the signal energy. While our analysis is based on a single cluster model, in Sec. \ref{numres} we demonstrate by simulation that the proposed scheme is robust also against multiple clusters.}
 For the single cluster model, we obtain
\begin{align}
\label{eq:channel_matrix}
\mathbf H = \sqrt{M_{\mathrm t} M_{\mathrm r} }h\;  \mathbf d_{\mathrm r}(\theta_{\mathrm r})\mathbf d_{\mathrm t}^H(\theta_{\mathrm t}),
\end{align}
where  $\mathbb E[|h|^2]=1/\ell(d)$, $\ell(d)$ denotes the path loss between BS and UE as a function of distance $d$, and $\bm \theta  = (\theta_{\mathrm t} , \theta_{\mathrm r} )$ is the single-cluster AoD/AoA pair.
{We assume that  $\bm \theta$ has prior joint distribution $f_0(\bm\theta)$ with support
$\supp(f_0){=}{\mathcal U}_{\mathrm t,0}\times  {\mathcal U}_{\mathrm r,0}$, which reflects the availability of prior AoD/AoA information acquired from previous beam-alignment phases, or based on geometric constraints (\emph{e.g.}, presence of buildings blocking the signal in certain directions). We assume that $h$ and $\bm \theta$ do not change over a frame, whose duration $T_{\mathrm{fr}}$ is chosen based upon the channel and beam coherence times $T_c$ and $T_b$ (time duration over which the AoD/AoA do not change appreciably)  to satisfy this property. In \cite{beam_coherence_time}, it has been reported that $T_c{\ll}T_b$. In the numerical values given below, $T_b{\sim}100T_c$.  
 Therefore, by choosing $T_{\mathrm{fr}}{\leq}T_c$, we ensure that the variations in $h$ and $\bm \theta$ over the frame duration $T_{\mathrm{fr}}$ are small and can be ignored. 
 For example, using the relationships of $T_c$ and $T_b$ in~\cite{beam_coherence_time}, we obtain $T_c{\simeq}10$[ms] and $T_b {\simeq}1$[s]
 for a UE velocity of 100[km/h]. In our numerical evaluations, we will therefore use $T_{\mathrm{fr}}{=}10$[ms]. It is noteworthy that this assumption has also been used  extensively in previous beam-alignment works, such as \cite{zorzi,alkhateeb,marzi}.} 
 
We assume that blockage occurs at {longer} time-scales than the frame duration,
determined by the geometry of the environment and mobility of users, {hence we neglect blockage dynamics within a frame duration} \cite{blockage}. 
 By replacing (\ref{eq:channel_matrix}) into (\ref{eq:receivedsignal}),
and defining the BS  and UE beam-forming  gains
$G_{\mathrm x}(\mathbf c_{\mathrm x},\theta_{\mathrm x})  =  M_{\mathrm x} |\mathbf d_{\mathrm x}^H(\theta_{\mathrm x}) \mathbf c_{\mathrm x}|^2,
\mathrm x\in\{\mathrm t,\mathrm r\}$,
 we get
\begin{align}
\label{eq:receivedsignal1}
 y =  h \sqrt{ P G_{\mathrm t}(\mathbf c_{\mathrm t},\theta_{\mathrm t})\cdot G_{\mathrm r}(\mathbf c_{\mathrm r},\theta_{\mathrm r})}
e^{j\Psi(\bm\theta)}   s + \hat{ w},
\end{align}
where $\hat{w}{\triangleq}\mathbf c_{\mathrm r}^H {\mathbf w} \sim \mathcal{CN}(0,N_0 W_{\mathrm{tot}})$ is the noise component and
 $\Psi(\bm\theta)=\angle{\mathbf d_{\mathrm t}^H(\theta_{\mathrm t}) \mathbf c_{\mathrm t}}
-\angle\mathbf d_{\mathrm r}^H(\theta_{\mathrm r})\mathbf c_{\mathrm r}
$ is the phase.

\label{page2}
\par In this paper, we use the \emph{sectored antenna} model \cite{sectored_model} to approximate the BS and UE beam-forming gains, represented in Fig.~\ref{fig:beam_pattern}. Under this model,
\begin{align}
\label{eq:gains}
G_{\mathrm x}(\mathbf c_{\mathrm x},\theta_{\mathrm x})\approx
G_{\mathrm x}(\mathcal B_{\mathrm x},\theta_{\mathrm x}) = \frac{2\pi}{\size{\mathcal B_{\mathrm x}}} \;\chi_{ \mathcal B_{\mathrm x}}\!\left(\theta_{\mathrm x}\right),\quad
\mathrm x\in\{\mathrm t,\mathrm r\},
\end{align}
where $\mathcal B_{\mathrm t}\subseteq (-\pi,\pi]$ is the range of AoD covered by $\mathbf c_{\mathrm t}$,  $\mathcal B_{\mathrm r}\subseteq (-\pi,\pi]$ is the range of AoA covered by $\mathbf c_{\mathrm r}$, $\chi_{\mathcal A}(\theta)$ is the indicator function of the event $\theta\in \mathcal A$, and $|\mathcal A|= \int_{\mathcal A} \mathrm d \theta$ is the measure of the set $\mathcal A$.
 Hereafter, the two sets $\mathcal B_{\mathrm t}$ and $\mathcal B_{\mathrm r }$ will be referred to as BS and UE beams, respectively.
 Additionally, we define $\mathcal B_k = \mathcal B_{\mathrm t,k} \times \mathcal B_{\mathrm r,k}$ as the 2-dimensional (2D) AoD/AoA support defined by the BS-UE beams.
 \label{page3}
{Note that the sectored model is used as an abstraction of the real model, 
which applies a precoding vector $\mathbf c_{\mathrm t}$ at the transmitter and 
a beamforming vector $\mathbf c_{\mathrm r}$ at the receiver. This abstraction, shown in Fig. \ref{fig:beam_pattern}, is adopted 
since direct optimization of $\mathbf c_{\mathrm t}$ and $\mathbf c_{\mathrm r}$
 is not analytically tractable, due to the high dimensionality of the problem. In Sec. \ref{numres} we show via Monte-Carlo simulation that,  by appropriate design of $\mathbf c_{\mathrm t}$ and $\mathbf c_{\mathrm r}$ to approximate the sectored model, 
our scheme attains near-optimal performance, and outperforms a state-of-the-art  bisection search scheme \cite{new_benchmark}; thus, the sectored antenna model provides  a valuable abstraction for practical design.
 Following the sectored antenna model, we obtain the received signal by replacing $G_{\mathrm x}(\mathbf c_{\mathrm x},\theta_{\mathrm x}) $ with $G_{\mathrm x}(\mathcal B_{\mathrm x},\theta_{\mathrm x})$ in \eqref{eq:receivedsignal1}, yielding 
\begin{align}
\label{eq:receivedsignal2}
 y =  h \sqrt{ P G_{\mathrm t}(\mathcal B_{\mathrm t},\theta_{\mathrm t})\cdot G_{\mathrm r}(\mathcal B_{\mathrm r},\theta_{\mathrm r})}
e^{j\Psi(\bm\theta)}   s + \hat{ w}.
\end{align}
Although the analysis in this paper is presented for ULAs (2D beamforming), the proposed scheme can be extended to the case of uniform planar arrays with 3D beamforming, by interpreting $\theta_{\mathrm{x}}, \mathrm{x} \in\{\mathrm{t},\mathrm{r}\}$ as a vector denoting the azimuth and elevation pair in $(-\pi,\pi]^2$ and the beam $\mathcal B_{\mathrm{x}} \subseteq (\pi,\pi]^2$. For notational convenience and ease of exposition, in this paper we focus on the 2D beamforming case (also adopted in, e.g., \cite{multi_res,new_benchmark,alkhateeb,bf_tradeoff}).}

\par The entire frame duration is split into two, possibly interleaved phases: a beam-alignment phase, whose goal is to detect the best beam to be used in the data communication phase.
To this end, we partition the slots $\mathcal I$ in each frame into the indices in the set $\mathcal {I}_{s}$, reserved for beam-alignment, and those  in the set $\mathcal {I}_{d}$, reserved for data communication, where $\mathcal I_s{\cap}\mathcal I_d{=}\emptyset$ and $\mathcal I_s{\cup}\mathcal I_d{=}\mathcal I$.
The optimal frame partition and duration of beam-alignment
are part of our design. 
{In the sequel, we describe the operations performed in the beam-alignment and data communication slots, and characterize their energy consumption.}

{\bf\underline{Beam-Alignment}:}
At the beginning of each slot $k{\in}\mathcal I_s$, the BS sends a beacon signal $\mathbf s$ of duration $T_B{<}T$
using the transmit beam $\mathcal B_{\mathrm{t},k}$ with power $P_k$,\footnote{{In practice, there are limits on how small the beacon duration can be made, due to peak power constraints \cite{shanon}, beacon synchronization errors \cite{rappaport_mmwave_book}, and auto-correlation properties of the beacon sequence \cite{rappaport_mmwave_book}.}}
 and the UE receives the signal using the receive beam $\mathcal B_{\mathrm r,k}$. Note that
$\mathcal B_{k}{=}\mathcal B_{\mathrm t,k}{\times}\mathcal B_{\mathrm r,k}$ and $P_k$ are design parameters.
If the UE detects the beacon (\emph{i.e.}, the
 AoD/AoA $\bm\theta$ is in $\mathcal B_{k}$, or a false-alarm occurs, see~\cite{mmnets17}),
 then, in the remaining portion of the slot  of duration $T{-}T_B$, it
 transmits an acknowledgment ($\mathrm{ACK}$) packet to the BS, denoted as $C_k{=}\mathrm{ACK}$. Otherwise (the UE does not detect the beacon due to either mis-alignment or misdetection error), it transmits $C_k{=}\mathrm{NACK}$.
We assume that the $\mathrm{ACK}/\mathrm{NACK}$ signal $C_k$ is received perfectly
and within the end of the slot
by the BS (for instance, by using a conventional microwave technology as a control channel \cite{Rangan}).

As a result of \eqref{eq:receivedsignal2},
the UE attempts to detect the beam, and generates the ACK/NACK signal based on
the following hypothesis testing problem,
\begin{align}
\label{eq:H0}
\hspace{-2.5mm}\mathcal{H}_1:\; &\mathbf y_k {=} \sqrt{N_0 W_{\mathrm{tot}}\nu_k} h  e^{j\Psi_k(\bm\theta)}\mathbf s{+} \hat{\mathbf w}_k,\hspace{-9mm}&\text{ (alignment, $\bm\theta{\in}\mathcal B_k$)}\\
\label{eq:H1}
\hspace{-2.5mm}\mathcal{H}_0:\; &\mathbf y_k {=}\hat{ \mathbf w}_k,\hspace{-10mm}&\text{ (misalignment, $\bm\theta{\notin}\mathcal B_k$)}
\end{align}
where $\mathbf y_k$ is the received signal vector,
{$\mathbf s$ is the transmitted symbol sequence,
  $\hat{\mathbf w}_k{\sim}\mathcal{CN}(\boldsymbol 0, N_0 W_{\mathrm{tot}} \mathbf I)$ is the AWGN vector, and  $\nu_k$
is related to the beam-forming gain in slot $k$,
\begin{align}
\label{nuk}
\nu_k = \frac{(2\pi)^2 P_k}{N_0 W_{\mathrm{tot}}\size{\mathcal B_{k}}}.
\end{align}
 {The optimal detector depends on the availability of prior information on $h$. We assume that 
  an estimate of the channel gain $\gamma{=}|h|^2$ is available at {the} BS and UE  at the beginning of each frame, denoted as
     $\hat\gamma{=}|\hat h|^2$, where $\hat h{=}h+e$ and $e{\sim}\mathcal{CN}(0,\sigma_{\mathrm{e}}^2)$ denotes {the} estimation noise.} A Neyman-Pearson threshold detector is optimal in this case,
\begin{align}
\frac{|\mathbf s^H \mathbf y_k|^2}{N_0 W_{\mathrm{tot}}\| \mathbf s\|_2^2} \mathop{\lessgtr}_{\mathcal H_1}^{\mathcal H_0} \tau_{\mathrm{th}}.
\end{align}
The detector's threshold  $\tau_{\mathrm{th}}$ and the transmission power $P_k$ are designed based on
the channel gain estimate $\hat \gamma$, so as to satisfy
 constraints on the false-alarm and misdetection probabilities, $p_{\mathrm{fa}},p_{\mathrm{md}}\leq p_e$. 
 We now compute these probabilities under the simplifying assumption that
 $\hat h$ and $e$ are independent, so that $h|\hat h\sim\mathcal{CN}(\hat h,\sigma_{\mathrm{e}}^2)$. Let
  $z_k \triangleq \frac{\mathbf s^H \mathbf y_k}{\sqrt{N_0 W_{\mathrm{tot}}}\| \mathbf s\|_2 }$, so that $|z_k|^2$ is the decision variable.  We observe that 
\begin{align*}
z_k = \begin{cases}
 \sqrt{\nu_k} h e^{j\Psi_k(\bm\theta)} \| \mathbf s \|_2  +\frac{\mathbf s^H\hat{\mathbf w}_k }{ \sqrt{N_0 W_{\mathrm{tot}}}\| \mathbf s\|_2}, &\text{ if $\mathcal H_1$ is true;} \\
  \frac{\mathbf s^H\hat{\mathbf w}_k }{ \sqrt{N_0 W_{\mathrm{tot}}}\| \mathbf s\|_2}, &\text{ if $\mathcal H_0$ is true.}
 \end{cases}
\end{align*}
Since
 $\hat h$ and $\hat{\mathbf w}_k$ are independent and $h=\hat h-e$, 
 we obtain
\begin{align}
\label{eq:zk_H1}
&\!f(z_k|\hat h,\mathcal H_1,\bm\theta)\!=\mathcal{CN}\!\left(\!\sqrt{\nu_k} \hat h  e^{j\Psi_k(\bm\theta)} \| \mathbf s \|_2, 1{+}{\nu_k}{\| \mathbf s \|_2^2 }\sigma_{\mathrm e}^2\right)\!,\!\!\\
\label{eq:zk_H0}
&f(z_k|\hat h,\mathcal H_0) = \mathcal{CN}\left(0,1\right),
\end{align}
so that $[|z_k|^2|\hat h,\mathcal H_0]\sim \mathrm{Exponential}(1)$, and the false-alarm probability can be expressed as
\begin{align*}
\label{eq:pfa}
p_{\mathrm{fa}}(\tau_{\mathrm{th}})\triangleq\mathbb P\left(|z_k|^2>\tau_{\mathrm{th}}\big|\hat h, \mathcal H_0\right)
= \exp\left(-\tau_{\mathrm{th}}\right). \numberthis
\end{align*}
Similarly, the misdetection probability is found to be
\begin{align*}
\label{eq:pmd}
&p_{\mathrm{md}}(\nu_k,\tau_{\mathrm{th}},\hat\gamma)\triangleq \mathbb P\left(|z_k|^2<\tau_{\mathrm{th}}\big|\hat h, \mathcal H_1\right)&
= 1-Q_1\left(\sqrt{\frac{2\hat\gamma\nu_k \| \mathbf s\|_2^2}{1+\nu_k\|\mathbf s\|_2^2\sigma_{\mathrm e}^2}},
\sqrt{\frac{2\tau_{\mathrm{th}}}{1+\nu_k\|\mathbf s\|_2^2\sigma_{\mathrm e}^2}}\right), \numberthis
\end{align*}
where  $Q_1(\cdot)$ is the first-order Marcum's Q function \cite{simon}. 
In fact, $z_k|(\hat h,\mathcal H_1)$ is complex Gaussian as in (\ref{eq:zk_H1}), so that, given $(\hat \gamma,\mathcal H_1)$,  
$\frac{2|z_k|^2}{{1+\nu_k\|\mathbf s\|_2^2\sigma_{\mathrm e}^2}}$  follows non-central chi-square distribution 
with 2 degrees of freedom and non-centrality parameter
$\frac{2\nu_k\hat\gamma\| \mathbf s \|_2^2}
{1+ {\nu_k}{\| \mathbf s \|_2^2 }\sigma_{\mathrm e}^2}$.

Herein, we design $\tau_{\mathrm{th}}$ and $P_k$ to achieve
$p_{\mathrm{fa}},p_{\mathrm{md}}
\leq p_{\mathrm e}$.
 To satisfy $p_{\mathrm{fa}}(\tau_{\mathrm{th}})\leq p_{\mathrm e}$ we need
\begin{align}
\label{eq:detect_th}
\tau_{\mathrm{th}}\geq -\ln\left(p_{\mathrm e}\right).
\end{align}
Since $Q_1(a,b)$ is an increasing function of $a{\geq}0$ and a decreasing function of $b{\geq}0$,
it follows that $p_{\mathrm{md}}(\nu_k,\tau_{\mathrm{th}})$ is a decreasing function of $\nu_k{\geq}0$ and an increasing function of $\tau_{\mathrm{th}}{\geq}0$. Then,
to guarantee $p_{\mathrm{md}}(\nu_k,\tau_{\mathrm{th}},\hat{\gamma}){\leq}p_{\mathrm e}$,
 \eqref{eq:detect_th} should be satisfied with equality to
attain the smallest $p_{\mathrm{md}}$; additionally, there exists $\nu^*{>}0$,
determined as the unique solution of $p_{\mathrm{md}}(\nu^*,\tau_{\mathrm{th}},\hat{\gamma}){=}p_{\mathrm e}$ and
 independent of the beam shape
$\mathcal B_{k}$, such that 
$p_{\mathrm{md}}(\nu_k,\tau_{\mathrm{th}},\hat{\gamma}){\leq}p_{\mathrm e}$ iff (if and only if) $\nu_k{\geq}\nu^*$.
Then, using \eqref{nuk} and letting
 $E_k{\triangleq}P_k T_{\mathrm{sy}} \|\mathbf s\|_2^2$ be the energy incurred for the transmission of the beacon $\mathbf s$ in slot $k$,
$E_k$ should satisfy
\begin{align}
\label{eq:ba_energy}
&E_{k}\geq \phi_s(p_{\mathrm e})\size{\mathcal B_{k}},
\\
&\text{where }
\label{phis}
\phi_s(p_{\mathrm{e}}){\triangleq}N_0 W_{\mathrm{tot}}\nu^*T_{\mathrm{sy}}
\|\mathbf s\|_2^2/(2\pi)^2
\end{align}
 is the $\text{energy}/\text{rad}^2$ required 
to achieve
 false-alarm and misdetection probabilities equal to $p_{\mathrm{e}}$.


\par 
{Note that false-alarm and misdetection errors are deleterious to performance, since they result in mis-alignment and outages during data transmission. Therefore, they should be minimized.
For this reason, in the first part of this paper we assume 
 that $p_{\mathrm{e}}{\ll}1$, and neglect the impact of these errors on beam-alignment.}
 Thus, we let $E_{k}{\geq}\phi_s\size{\mathcal B_{k}}$
be the energy required in each beam-alignment slot to guarantee detection with high probability,  where $\phi_s$ is computed under some small $p_{\mathrm{e}}\ll 1$.
We will consider the impact of these errors in Sec.~\ref{sec:fa_md_impact}.\footnote{The design of beam-alignment schemes robust to errors when $p_{\mathrm{e}}\not\ll 1$ has been considered in \cite{allerton2018}. Its analysis is outside the scope of this paper.}

{\bf\underline{Data Communication}:}
In the communication slots indexed by ${k\in}\mathcal I_d$, the BS uses $\mathcal B_{\mathrm t, k}$, rate $R_k$, and transmit power $P_k$, while the UE processes the received signal using the beam $\mathcal B_{\mathrm r, k}$. Therefore, 
letting $\gamma=|h|^2$ and 
$\nu_k$ as in \eqref{nuk}.
the instantaneous SNR can be expressed as
\begin{align}
 \nonumber
{\mathsf{SNR}}_k&=\frac{ \gamma P_k G_{\mathrm t}(\mathcal B_{\mathrm t,k},\theta_{\mathrm t}) G_{\mathrm r}(\mathcal B_{\mathrm r,k},\theta_{\mathrm r}) }{ N_0 W_{\mathrm{tot}}}\\
\label{eq:snr2}
&=  \nu_k\gamma \chi_{\mathcal B_{\mathrm t,k}}(\theta_{\mathrm t})\chi_{\mathcal B_{\mathrm r,k}}(\theta_{\mathrm r}).
\end{align}
\par Outage occurs if $W_{\mathrm{tot}}\log_2(1{+}\mathsf{SNR}_k){<}R_k$ due to either mis-alignment between transmitter and receiver, or low channel gain $\gamma$.
The probability of this event, $p_{out}$, can be inferred from the posterior probability distribution of the
AoD/AoA pair $\bm \theta $ and the channel gain $\gamma$, given its estimate $\hat\gamma$, and
the history of BS-UE beams and feedback until slot $k$, denoted as $\mathcal H^k\triangleq\{(\mathcal B_{0},C_0),\ldots,(\mathcal B_{k-1},C_{k-1})\}$.
Thus, $p_{out}\triangleq \mathbb{P}(W_{\mathrm{tot}}\log_2(1+{\mathsf{SNR}}_k) < R_k | \hat\gamma, \mathcal H^k)$, yielding
\begin{align*}
p_{out}
&\stackrel{(a)}{=} \mathbb{P}\left({\mathsf{SNR}}_k < 2^{\frac{R_k}{W_{\mathrm{tot}}}}-1 | \hat\gamma, \bm \theta \in \mathcal B_k\right) \mathbb{P}(\bm \theta \in \mathcal B_k | \mathcal H^k) \\&\qquad + \mathbb{P}(\bm \theta \not \in \mathcal B_k | \mathcal H^k) \\
&\stackrel{(b)}{=} 1 -\bar{F}_{\gamma}\left(\frac{2^{\frac{R_k}{W_{\mathrm{tot}}}}-1}{\nu_k}\bigg|\hat\gamma\right)\mathbb{P}(\bm\theta \in \mathcal B_{k}|  \mathcal H^k), \numberthis
\end{align*}
where  (a) follows from {the law of} total probability and $\mathbb{P}(\bm \theta{\in}\mathcal B_k | \mathcal H^k) $ denotes the probability of correct beam-alignment; (b) follows  by substituting $\bar{F}_{\gamma}(x|\hat\gamma){\triangleq}\mathbb{P}(\gamma{\geq}x | \hat\gamma)$ into (a),
given as
\begin{align}
\bar{F}_{\gamma}(x|\hat\gamma) = Q_1\left(\sqrt{{2\hat\gamma}/{\sigma_e^2}},\sqrt{{2x}/{\sigma_e^2}} \right) .
\end{align}
Herein, we use the notion of $\epsilon$-outage capacity to design 
 $R_k$, defined as the largest transmission rate such that $p_{out}{\leq}\epsilon$, 
 for a target outage probability $\epsilon{<}1$. This can  be expressed as
\begin{align*} 
\label{eq:outage_capacity}
&{\mathcal C_{\epsilon}(P_k,\mathcal B_k|\mathcal H^k,\hat\gamma)}
\triangleq  W_{\mathrm{tot}}\log_2 \left( 1+\nu_k \bar{F}_{\gamma}^{-1}\left(\left.\frac{1-\epsilon}{\mathbb{P}(\bm\theta{\in}\mathcal B_k |  \mathcal H^k)} \right|\hat\gamma\right)\right), \numberthis
 \end{align*}
where $\bar{F}_{\gamma}^{-1}(\cdot|\hat\gamma)$ denotes the inverse posterior CCDF of $\gamma$, conditional on $\hat\gamma$. In other words, if 
 $R_k{\leq}{\mathcal C_{\epsilon}(P_k,\mathcal B_k|\mathcal H^k,\hat\gamma)} $, then
{the transmission is successful with probability at least $1{-}\epsilon$, and}
the average rate is at least $(1{-}\epsilon)R_k$. 
{Note that, in order to achieve the target $p_{out}{\leq}\epsilon$, the probability of correct beam-alignment must satisfy 
$\mathbb{P}(\bm \theta{\in}\mathcal B_k | \mathcal H^k){\geq}1{-}\epsilon$.} This
 can be achieved with a proper choice of $\mathcal B_k$, as discussed next.

Since the ACK/NACK feedback after data communication is generated by higher layers (\emph{e.g.,} network or transport layer), we do not use it to improve beam-alignment. We define $C_k{=}\mathrm{NULL}$, $\forall k{\in}\mathcal I_d$, to distinguish it from the ACK/NACK feedback signal in the beam-alignment slots.


\section{Problem Formulation}
\label{sec:probform}
In this section, we formulate the {optimization} problem, {and characterize it as a} Markov decision process (MDP). 
\label{page1xx} {The goal is to minimize the power consumption at the BS over a frame duration, while achieving the quality of service (QoS) requirements of the UE (rate and delay).  Therefore, the objective function of the following optimization problem captures the beam-alignment and data communication energy costs; the QoS requirements are specified in the constraints through a rate requirement $R_{\min}$ of the UE along with an outage probability of $\epsilon$; additionally, the frame duration $T_{\mathrm{fr}}$ represents a delay guarantee on data transmission.}
The design variables in slot $k$ are denoted by the  4-tuple
${\mathbf a}_k= (\xi_k, P_k, \mathcal B_{k}, R_k)$, where $\xi_k$ corresponds to the decision  of
whether to perform beam-alignment ($\xi_k{=}1$) or data communication ($\xi_k{=}0$);
 we let $R_k{=}0$ for beam-alignment slots ($\xi_k=1$). 
  With this choice of $\mathbf a_k$, we aim to optimally select the  beam-alignment slots $\mathcal I_s$ and data communication slots $\mathcal I_d$. If  a slots is selected for beam-alignment ($\xi_k=1$), we aim to optimize the associated power $P_k$ and 2D beam $\mathcal B_k$. Likewise, if a slot is selected for data communication ($\xi_k=0$), we aim to optimize the associated power $P_k$, data rate $R_k$, and 2D beam $\mathcal B_k$.
 Mathematically, the optimization problem is stated as
\begin{align}
\label{eq:op_P0} 
\mathrm{P}_1:&\quad
\bar P\triangleq
\min_{\mathbf a_0,\ldots,\mathbf a_{N-1}}\quad \frac{1}{T_{\mathrm{fr}}}\mathbb{E}\left[\sum_{k=0}^{N-1} E_{k}\bigg|f_0 \right]\\ \nonumber
\textnormal{s.t.}\quad &{\mathbf a}_k {=} (\xi_k, P_k, \mathcal B_{k}, R_k), \forall k,
\\
& 
\label{eq:opt_P0_c4}
\mathcal B_k{=} \mathcal B_{\mathrm t,k} {\times} \mathcal B_{\mathrm r,k}\subseteq[-\pi,\pi]^2, \forall k, \\
\label{eq:opt_P0_c1}
& E_{k} \geq \phi_s|\mathcal B_k |,\ \forall k\in\mathcal I_s,
\\
\label{eq:opt_P0_c2}
&\frac{1}{N}\sum_{k} R_k{\geq}R_{\min},\ R_k{\leq}\mathcal C_\epsilon(P_k,\mathcal B_k|\mathcal H^k,\hat\gamma),\forall k{\in}\mathcal I_d,
\\
\label{eq:opt_P0_c5}
&P_k=E_k/[\xi_kT_B+(1-\xi_k)T],\ \forall k,
\end{align}
where $f_0$ in (\ref{eq:op_P0}) denotes the prior belief over $\bm\theta$; 
 \eqref{eq:opt_P0_c4} defines the 2D beam $\mathcal B_k$;
(\ref{eq:opt_P0_c1}) gives the energy consumption in the beam-alignment slots;
 (\ref{eq:opt_P0_c2}) ensures the rate requirement $R_{\min}$ over the frame, and that $ R_k$ is within the $\epsilon$-outage capacity, see (\ref{eq:outage_capacity}); 
 (\ref{eq:opt_P0_c5}) gives the relation between energy and power.\footnote{Data communication takes the entire slot, whereas beam-alignment occurs over a portion $T_B<T$ of the slot to allow for the time to receive the ACK/NACK feedback from the receiver.}
Since the cost is the average BS power consumption, the inequality constraints (\ref{eq:opt_P0_c1})-(\ref{eq:opt_P0_c2})
 must be tight, i.e.,
we replace them with
\begin{align}
\label{eq:opt_P1_c1}
&E_{k} = \xi_k\phi_s \size{\mathcal B_k}
+(1-\xi_k)
\frac{{\psi}_d(R_k)\size{\mathcal B_k}}{\bar{F}_{\gamma}^{-1}\left(\frac{1-\epsilon}{\mathbb{P}(\bm{\theta} \in \mathcal B_k | \mathcal H^k)} \big|\hat\gamma \right)},\\
\label{eq:opt_P1_c2}
&\frac{1}{N}\sum_{k}R_k = R_{\min},
\end{align}
where \eqref{eq:opt_P1_c1} when $\xi_k{=}0$ is obtained by inverting \eqref{eq:opt_P0_c2} via \eqref{eq:outage_capacity} and \eqref{nuk} (with equality)
to find $P_k$ and $E_k=P_kT$,
and we have defined {the $\text{energy}/\text{rad}^2$ required to achieve the rate $R$}
$${\psi}_d(R) \triangleq (2\pi)^{-2}N_0 W_{\mathrm{tot}} T  (2^{\frac{R}{W_{\mathrm{tot}}}}-1).$$
Hereafter, we exclude $P_k$ from the design space, since it is uniquely defined by the set of equality constraints (\ref{eq:opt_P0_c5})-(\ref{eq:opt_P1_c1}). Thus, we simplify the design variable to ${\mathbf a}_k=  (\xi_k,\mathcal B_{k}, R_k)$.

We pose $\mathrm P_1$ as an MDP \cite{bertsekas} over the time horizon $\mathcal I$. 
The state at the start of slot $k$ is $(f_k,D_k)$, where $f_k$ is the probability distribution over the AoD/AoA pair $\bm{\theta}$, given the history $\mathcal H^k$ up to slot $k$,  denoted as \emph{belief}; $D_k$ is
 the backlog (untransmitted data bits).
 Initially, $f_0$ is the prior belief and 
$D_0{\triangleq}R_{\min}T_{\mathrm{fr}}$.
Given $(f_{k},D_{k})$, the BS and UE select $\mathbf a_{k}{=}\left(\xi_{k},\mathcal B_{k},R_{k}\right)$.\footnote{{Since feedback is error-free, both BS and UE have the same information to generate the action $\mathbf a_{k}$ and their beams.}}
Then, the UE generates the feedback signal: if $\xi_{k}{=}0$ (data communication), then $C_{k}{=}\mathrm{NULL}$;
if $\xi_{k}{=}1$ (beam-alignment), then 
$C_{k}{=}\mathrm{ACK}$ if $\bm{\theta}{\in}\mathcal B_k$, with probability
\begin{align}
\label{pack}
\mathbb P(C_{k}=\mathrm{ACK}|f_{k},\mathbf a_{k})=
\int_{\mathcal{B}_{k}} f_{k}(\bm \theta )d \bm \theta,
\end{align}
and $C_{k}{=}\mathrm{NACK}$ otherwise.
Upon receiving $C_{k}$, the new
backlog in slot $k+1$ becomes\footnote{If $D_{k+1}\leq 0$, all bits have been transmitted.}
\begin{align}
D_{k+1}=\max\bigr\{D_{k}-R_{k}T,0\bigr\},
\end{align}
and the new belief $f_{k+1}$ is computed via Bayes' rule, as given in the following lemma.
\begin{lemma}
\label{lemma:sufficient_statistcs}
Let $f_0$ be the prior belief on $\bm \theta$ with support $\supp(f_0) = {\mathcal U}_0$.
Then, 
\begin{align}
\label{eq:belief_construct1}
f_{k}(\bm \theta) = \frac{f_0(\bm \theta)}{\int_{{\mathcal U}_{k} } f_0(\tilde{\bm \theta}) \mathrm  d \tilde{\bm \theta}}  \chi_{{\mathcal U}_{k}}(\bm \theta),
\end{align}
where ${\mathcal U}_{k}\triangleq\supp(f_{k})$ is updated recursively as
\begin{align}
\label{eq:supp_update}
{\mathcal U}_{k+1} =\left\{\begin{array}{ll}
{\mathcal U}_k \cap \mathcal B_k, &k \in \mathcal I_{\mathrm s}, C_k=\mathrm{ACK}\\
{\mathcal U}_k \setminus{\mathcal B_k }, &k \in \mathcal I_{\mathrm s}, C_k=\mathrm{NACK}\\
{\mathcal U}_k, &k \in \mathcal I_{\mathrm d}.
\end{array}\right.
\end{align}
\end{lemma}
\begin{proof}
The proof follows by induction using Bayes' rule.
In fact, if $C_{k}{=}\mathrm{ACK}$ in a beam-alignment slot, then it can be inferred that
$\bm \theta{\in}\mathcal U_k{\cap}\mathcal B_k$; otherwise  ($C_{k}{=}\mathrm{NACK}$) the UE lies outside 
$\mathcal B_k$, but within the support of $f_k$, i.e., $\bm\theta{\in}\mathcal U_k{\setminus}\mathcal B_k$.
In the data communication slots, no feedback is generated, hence
$f_{k+1}{=}f_k$ and ${\mathcal U}_{k+1}{=}{\mathcal U}_{k}$.
 A detailed proof is given in Appendix \ref{proofoflemlemma:sufficient_statistcs}.
\end{proof}
\vspace{-2mm}
Lemma~\ref{lemma:sufficient_statistcs} implies that
$\mathcal U_k$ is a sufficient statistic for decision making in slot $k$, and is
 updated recursively via (\ref{eq:supp_update}).
 Accordingly, the state space is defined as
\begin{align}
\mathcal S\equiv\{(\mathcal U,D):\mathcal U\subseteq\mathcal U_0,0\leq D\leq D_0\}.
\end{align}

\par Given the data backlog $D_k{=}D$, the action space is expressed as\footnote{Note that, for a data communication action
$(0,\mathcal B, R)$, we assume that $R>0$; in fact, data communication with zero rate is equivalent to a
beam-alignment action $(1,\emptyset,0)$ with empty beam.}
\begin{align}\nonumber
\label{eq:action_space}
\mathcal A(D) \equiv &\left\{(0,\mathcal B, R): \mathcal B\equiv\mathcal B_{\mathrm t}\times\mathcal B_{\mathrm r}\subseteq [-\pi,\pi]^2, 0{<}R{\leq}D/T \right\} 
\\&\cup \left\{(1,\mathcal B, 0): \mathcal B\equiv\mathcal B_{\mathrm t}\times\mathcal B_{\mathrm r} \subseteq [-\pi,\pi]^2 \right\}.
\end{align}
Given $({\mathcal U}_{k},D_k)\in\mathcal S$,
the action $\mathbf a_k{\in}\mathcal A(D_k)$ is chosen based on policy $\mu_k$,
which determines the BS-UE beam  $\mathcal B_k$ and
whether to perform beam-alignment ($\xi_k{=}1$, $R_k{=}0$) or  data communication ($\xi_k{=}0$, $R_k{>}0$),
with energy cost $E_k$ given by \eqref{eq:opt_P1_c1}. With this notation,
we can express the problem $\mathrm P_1$ as that of finding the policy $\mu^*$ which minimizes the power consumption under rate requirement and outage probability constraints,
 \begin{align*}
 \label{P2}
& \mathrm{P}_2: \quad
\bar P\triangleq
\min_{\mu} \; \frac{1}{T_{\mathrm{fr}}}\mathbb{E}_{\mu} \Biggl[ \sum_{k=0}^{N-1}
c(\mathbf a_k;{\mathcal U}_{k},D_k)
\bigg|{\mathcal U}_{0},D_0,f_0\Biggr],\\
 &\mathrm{s.t.}\qquad
 D_{k+1} = D_k - T R_k, \forall k \in \mathcal I, \; \;  D_{N} = 0, \numberthis
 \end{align*}
 where
 we have defined the cost per stage in state 
$({\mathcal U}_{k},D_k)$ under action $\mathbf a_k$ as
 \begin{align*}
\label{eq:cost1}
 c(\mathbf a_k;{\mathcal U}_{k},D_k){=}
\left[\xi_k\phi_s
{+}\frac{(1{-}\xi_k)\psi_d(R_k)}{\bar{F}_{\gamma}^{-1}\left(\frac{1-\epsilon}{\mathbb{P}(\bm \theta \in \mathcal B_{k} |\mathcal U_{k})}\big|\hat\gamma \right)}\right]
\size{\mathcal B_{k}}
, \numberthis
\end{align*}
and we used the sufficient statistic (Lemma \ref{lemma:sufficient_statistcs})
to express
 $
 \mathbb{P}(\bm \theta{\in}\mathcal B_k |\mathcal H^k)
{=}
\mathbb{P}(\bm \theta{\in}\mathcal B_k |\mathcal U_k)
$ in \eqref{eq:opt_P1_c1}.  
 $\mathrm{P}_2$ can be solved via dynamic programming (DP):
the value function
in state $({\mathcal U}_k , D_k)$ under action $\mathbf a_k{\in}\mathcal A(D_k) $,
$V_k( \mathbf a_k; {\mathcal U}_k,D_k)$,
 and the optimal value function, $V_k^*({\mathcal U}_k,D_k) $, are 
 expressed as
\begin{align}\nonumber
V_k( \mathbf a_k; {\mathcal U}_k,D_k) &=  c( \mathbf a_k; {\mathcal U}_k,D_k)\\ \nonumber
&+\mathbb{E}\left[ V_{k+1}^{*}( {\mathcal U}_{k+1},D_{k+1})\bigg|{\mathcal U}_k,D_k; {\mathbf a}_k\right],\\
\label{eq:cost_to_go_iter}
V_k^*({\mathcal U}_k,D_k) 
&= \min_{\mathbf a_k\in\mathcal A(D_k) } 
V_k( \mathbf a_k; {\mathcal U}_k,D_k),
\end{align}
where 
the minimum is attained by the optimal policy.
To enforce $D_N{=}0$, we initialize it as
\begin{align}
\label{ctgo_opt_N}
V_N^{*}(  {\mathcal U}_N,D_N)=
\begin{cases}
0, &D_N = 0\\
\infty, & D_N>0.
\end{cases}
\end{align}

Further analysis  is not doable for a generic prior $f_0$. 
To unveil structural properties, we proceed as follows:
\begin{enumerate}[leftmargin=.60cm]
\item We optimize over the extended action space
\begin{align}
\label{eq:extended_action_space}
 {\mathcal A}_{\mathrm{ext}}(D) \equiv &\left\{(0,\mathcal B, R): \mathcal B \subseteq [-\pi,\pi]^2, 0 < R \leq D/T \right\} \cup \left\{(1,\mathcal B, 0): \mathcal B \subseteq[-\pi,\pi]^2\right\},
\end{align}
obtained by removing the "rectangular beam" constraint
$\mathcal B{\equiv}\mathcal B_{\mathrm t}{\times}\mathcal B_{\mathrm r}$ 
in \eqref{eq:action_space}.
Thus, $\mathcal B\in{\mathcal A}_{\mathrm{ext}}(D)$ can be any subset of $[-\pi,\pi]^2$, not restricted to a ``rectangular'' shape $\mathcal B\equiv\mathcal B_{\mathrm t,k}\times\mathcal B_{\mathrm r,k}$.
By optimizing over an extended action space, a lower bound to the value function is obtained, denoted as $\hat V_k^*({\mathcal U}_k,D_k)\leq V_k^*({\mathcal U}_k,D_k)$,
possibly not achievable by a "rectangular" beam.
\item In Sec. \ref{sec:unifprior}, we find structural properties under such extended action space, for the case of a uniform belief $f_0$. In this setting, we prove the optimality of a \emph{fractional search} method, which selects
$\mathcal B_k $ as $\mathcal B_k \subseteq\mathcal U$ with $|\mathcal B_k |=\rho_k|\mathcal U_k|$ (beam-alignment)
or $|\mathcal B_k |=\vartheta|\mathcal U_k|$ (data communication), for appropriate \emph{fractional parameters} $\rho_k$ and $\vartheta$; additionally, we prove the optimality of a \emph{deterministic} duration of the beam-alignment phase (Theorems   \ref{th:opt_comm_beams} and \ref{optimalpoliuniform}).
\item 
 In Sec. \ref{sec:rec_beams_optimalality}, we prove that
 such lower bound is indeed achievable by
 a \emph{decoupled fractional search} method, which decouples
 the BS and UE beam-alignment over time using rectangular beams, hence it is optimal.
\item In Sec. \ref{sec:nonunifprior}, we use these results to design a heuristic policy with performance guarantees for the case of non-uniform prior $f_0$, and show that the uniform prior is the worst case.
\end{enumerate}

\vspace{-5mm}
\section{Uniform Prior}
\label{sec:unifprior}
We denote the beam $\mathcal B$ taking value from the extended action space
${\mathcal A}_{\mathrm{ext}}(D)$ as "2D beam", to distinguish it from 
$\mathcal B{\in}{\mathcal A}(D)$, that obeys a "rectangular" constraint. 
Additionally, since the goal is to minimize the energy consumption, we restrict $\mathcal B{\subseteq}\mathcal U$ during data communication and
$\mathcal B{\subset}\mathcal U$ during beam-alignment,
yielding the following extended action space in state $(\mathcal U,D)$:\footnote{In fact, the AoD/AoA lie within the belief support
$\mathcal U_k$; projecting a "2D beam" outside of $\mathcal U_k$ is suboptimal, since it yields an unnecessary energy cost.
Additionally, choosing $\mathcal B_k{=}\mathcal U_k$ during beam-alignment is suboptimal, since it triggers an ACK with probability one,
which is uninformative; we thus restrict $\mathcal B_k\subset\mathcal U_k$.
A formal proof is provided in Appendix  \ref{proofofoptbeam}.}
\begin{align}
\nonumber
{\mathcal A}_{\mathrm{ext}}(\mathcal U,D) \equiv &\left\{(0,\mathcal B, R): \mathcal B \subseteq\mathcal U, 0 < R \leq D/T \right\}\\
\label{eq:extended_action_space2}
&\cup \left\{(1,\mathcal B, 0): \mathcal B \subset\mathcal U\right\}.
\end{align}

In this section, we consider the independent uniform prior on $\bm \theta = (\theta_{\mathrm t},\theta_{\mathrm r}) $, i.e.,
\begin{align}
f_0(\bm \theta) = 
f_{\mathrm r,0}(\theta_{\mathrm r}) \cdot f_{\mathrm t,0}(\theta_{\mathrm t}), \ 
f_{\mathrm x,0}(\theta_{\mathrm x}) =
\frac{\chi_{{\mathcal U}_{\mathrm x,0}}(\theta_{\mathrm x})}{\size{{\mathcal U}_{\mathrm x,0}}}.
\end{align}
From Lemma \ref{lemma:sufficient_statistcs}, it directly follows that $f_k$ is \emph{uniform} in its support $\mathcal U_k$,
and the state transition probabilities from state $({\mathcal U}_k,D_k)$ under the beam-alignment action $(1,\mathcal B_k,0)\in{\mathcal A}_{\mathrm{ext}}(\mathcal U,D)$, given in \eqref{pack} for the general case,
 can be specialized as $D_{k+1}=D_k$ and
\begin{align*}
\label{eq:BS_trans_prob}
{\mathcal U}_{k+1}=
\begin{cases}
\mathcal B_{k},
&\text{w.p. }
\frac{|\mathcal B_{k}|}{|{\mathcal U}_{k}|},\\
{\mathcal U}_{k}\setminus \mathcal B_{k},
&\text{w.p. }
1-\frac{|\mathcal B_{k}|}{|{\mathcal U}_{k}|},
\end{cases}\numberthis
\end{align*}
 where ``w.p.'' abbreviates ``with probability''. 
On the other hand, under the data communication action 
$(0,\mathcal B_k,R_k)$, the new state  becomes
$\mathcal U_{k+1} =\mathcal U_{k} $, and $D_{k+1} = D_k - R_k T$.  

In order to determine the optimal policy with extended action set, we proceed as follows:
\begin{enumerate}[leftmargin=.60cm]
\item In Sec. \ref{s1}, we find the structure of the optimal data communication beam, as a function of the transmit rate 
$R_k$ and support $\mathcal U_k$, and investigate its energy cost;
\item Next, in Sec. \ref{s2}, we prove that it is suboptimal to perform beam-alignment \emph{after} data communication within the frame. Instead, it is convenient to narrow down the beam as much as possible via beam-alignment, to achieve the most energy-efficient data communication;
\item Finally, in Sec. \ref{s3}, we investigate the structure of the value function, to prove
the optimality of a \emph{fixed-length} beam-alignment and of a \emph{fractional-search method}.
\end{enumerate}
\vspace{-2mm}
\subsection{Optimal data communication beam}
\label{s1}
In the following theorem, we find the optimal 2D beam for data communication.
\begin{theorem}
\label{th:opt_comm_beams}
 In any communication slot $k\in \mathcal I_d$, the 2D beam $\mathcal B_k$ is optimal iff
\begin{align}
\label{eq:optimal_comm_beams}
\mathcal B_k \subseteq \mathcal U_k \; \; 
&\size{\mathcal B_k} = \vartheta \size{\mathcal U_k},
\end{align}
where $\vartheta = (1-\epsilon)/q^*$, with $q ^*= \arg\max_{q\in [1-\epsilon,1]} q \bar F_{\gamma}^{-1}(q|\hat\gamma)$.
\end{theorem}
\begin{proof}
The proof is provided in Appendix \ref{app:proof_of_opt_comm}.
\end{proof}
{The significance of this result is that the optimal beam in the data communication phase is a \emph{fraction} $\vartheta$ of the region of uncertainty $\mathcal U_k$, with $\vartheta$ reflecting the desired outage constraint.} By substituting (\ref{eq:optimal_comm_beams})  into \eqref{eq:cost1},
and  letting
\begin{align}
\label{phid}
\phi_d(R,\epsilon)\triangleq\frac{{\psi_d(R)}{(1-\epsilon)}}{q^*{\bar F}_{\gamma}^{-1}(q^*|\hat\gamma)}
\end{align}
 be the energy/$\mathrm{rad}^2$ to achieve transmission rate $R$ with outage probability $\epsilon$,
the cost per stage of a data communication action with beam given by Theorem \ref{th:opt_comm_beams} can be expressed as
 \begin{align*}
\label{eq:comm_val_ftn_final}
 c(\mathbf a_k;{\mathcal U}_k , D_k)&=
\phi_d(R_k,\epsilon) \size{{\mathcal U}_k}. \numberthis
\end{align*}
\vspace{-10mm}
\subsection{Beam-alignment before data communication is optimal}
\label{s2}
{
In Theorem \ref{th:sens_comm},
 we prove that it is suboptimal to precede data communication to beam-alignment. Instead, it is more energy efficient to narrow down the beam as much as possible via beam-alignment, before switching to data communication.}
\vspace{-2mm}
\begin{theorem}
\label{th:sens_comm}
Let $\mu$ be a policy and $\{({\mathcal U}_{k},D_k),k\in \mathcal I \}$ be a realization of the state process under $\mu$ such that $\exists j:\xi_j({\mathcal U}_j,D_j) = 0$ and $\xi_{j+1}({\mathcal U}_{j+1},D_{j+1})= 1$ (beam-alignment is followed by data communication, for some slot $j$). Then, $\mu$ is suboptimal.
\end{theorem}
\begin{proof}
{The theorem is proved in two parts using contradiction. The first part deals with the case when a data communication slot is followed by a beam-alignment slot having non-zero beam-width. The second part deals with the case when a data communication slot is followed by a beam-alignment slot having zero beam-width.}
Let $\mu$ be a policy such that, for some state
 $({\mathcal U}_j, D_j)$ and slot index $j$, $\mu_j
({\mathcal U}_j,D_j) =(0,{\mathcal B}_j,R_j)$,
satisfying the conditions of Theorem \ref{th:opt_comm_beams} (data communication action);
thus, the state at $j+1$ is $({\mathcal U}_{j+1}, D_{j+1})=
({\mathcal U}_j,D_j-TR_j)
$.
Further, assume that, in this state, $\mu_{j+1}({\mathcal U}_j, D_j-TR_j)
=(1,{\mathcal B}_{j+1}, 0)$
(beam-alignment), 
with ${\mathcal B}_{j+1}\subset\mathcal U_j$ (strict subset, see \eqref{eq:extended_action_space2}),
so that the state in slot $j+2$ is either
$({\mathcal B}_{j+1}, D_j-TR_j)$ with probability 
$\size{{\mathcal B}_{j+1}}/\size{{\mathcal U}_j}$ (ACK), or
$({\mathcal U}_j\setminus{\mathcal B}_{j+1}, D_j-TR_j)$ otherwise (NACK).
This policy follows beam-alignment to data communication, and we want to prove that it is suboptimal.
We use
 (\ref{eq:cost_to_go_iter}) to get the cost-to-go function in slot $j$ under policy $\mu$ as
\begin{align*}
\label{eq:ctg_j_mu}
V_j^{\mu}({\mathcal U}_j ,D_j)
&=
 \phi_d(R_j,\epsilon)\size{{\mathcal U}_j }+  V_{j+1}^{\mu}({\mathcal U}_j,D_j-T R_j)
\\&
=\phi_d(R_j,\epsilon)\size{{\mathcal U}_j } + \phi_s \size{\mathcal B_{j+1}} 
+ \frac{\size{\mathcal B_{ j+1}}}{\size{{\mathcal U}_j}}V_{j+2}^{\mu}(\mathcal B_{j+1},D_j-T R_j)
\\&\;\;\;+ 
\frac{\size{
{\mathcal U}_j\setminus
{\mathcal B_{j+1}}}}{\size{{\mathcal U}_j}}V_{j+2}^{\mu}({\mathcal U}_j\setminus{\mathcal B_{ j+1}},D_j-TR_j).
\numberthis
\end{align*}
We consider the two cases $|\mathcal B_{j+1}|{>}0$ and $|\mathcal B_{j+1}|{=}0$ separately.
In both cases, we will construct a new policy $\tilde \mu$ and compare the cost-to-go function at $j$ under the two policies $\mu$ and $\tilde \mu$.

$\underline{|\mathcal B_{j+1}| > 0}$:
We define $\tilde \mu$ 
as being equal to $\mu$ except for the following:
$\tilde\mu_j(\mathcal U_j,D_j)
{=}(1,\mathcal B_{j+1},0)
$, so that 
$\tilde \mu$ executes the beam-alignment action in slot $j$, instead of $j{+}1$.
It follows that 
\begin{align*}
\label{eq:ctg_j_tildemu}
V_j^{\tilde\mu}(\mathcal U_j,D_j)
=&\phi_s \size{\mathcal B_{j+1}} + \frac{\size{\mathcal B_{j+1}}}{\size{{\mathcal U}_j}}V_{j+1}^{\tilde\mu}(\mathcal B_{j+1},D_j)+ 
\frac{\size{{\mathcal U}_j\setminus{\mathcal B_{j+1}}}}{\size{{\mathcal U}_j}}V_{j+1}^{\tilde\mu}({\mathcal U}_j\setminus{\mathcal B_{j+1}},D_j).
\numberthis
\end{align*}
Furthermore, we design $\tilde\mu$ such that
$\tilde\mu_{j+1}({\mathcal B}_{j+1},D_j)
=(0,\tilde{\mathcal B}_{j+1}^\prime,\,R_j)$
and 
$\tilde\mu_{j+1}({\mathcal U}_j\setminus
{\mathcal B}_{j+1},D_j)
{=}(0,\tilde{\mathcal B}_{j+1}^{\prime\prime},R_j)$, so that 
$\tilde \mu$ executes the data communication action in slot $j+1$, instead of $j$,
with beams $\tilde{\mathcal B}_{j+1}^\prime$ and $\tilde{\mathcal B}_{j+1}^{\prime \prime}$ satisfying the conditions of Theorem \ref{th:opt_comm_beams}.
It follows that the system moves from state
$({\mathcal B}_{j+1},D_j)$
to $({\mathcal B}_{j+1},D_j-TR_j)$,
and from $({\mathcal U}_{j}\setminus
{\mathcal B}_{j+1},D_j)$
to $({\mathcal U}_{j}\setminus
{\mathcal B}_{j+1},D_j-TR_j)$
under policy $\tilde\mu$, yielding
\begin{align*}
\label{eq:ctg_j_tildemu2} 
&V_{j+1}^{\tilde\mu}({\mathcal B}_{j+1},D_j)
\stackrel{(a)}{=}
\phi_d(R_j,\epsilon)\size{{\mathcal B}_{j+1}}
+V_{j+2}^{\tilde\mu}({\mathcal B}_{j+1},D_j-T R_j),
\\
&V_{j+1}^{\tilde\mu}(
{\mathcal U}_j\setminus\mathcal B_{j+1}, D_j)
\stackrel{(b)}{=}
\phi_d(R_j,\epsilon)\size{{\mathcal U}_{j}\setminus\mathcal B_{j+1}}
 +V_{j+2}^{\tilde\mu}({\mathcal U}_j\setminus\mathcal B_{j+1},D_j-T R_j). \numberthis
\end{align*}
By substituting (\ref{eq:ctg_j_tildemu2})-(a),(b) into \eqref{eq:ctg_j_tildemu},
and using the fact that $\tilde\mu_k$ and $\mu_k$ are identical for $k\geq j+2$
(hence $V_{j+2}^{\tilde\mu}=V_{j+2}^{\mu}$), it follows that
\begin{align*}
\label{eq:comp1}
V_j^{\tilde\mu}({\mathcal U}_j,D_j)
-V_j^{\mu}({\mathcal U}_j,D_j)
&\stackrel{(a)}{=}
-\phi_d(R_j,\epsilon)2\frac{\size{\mathcal B_{j+1}}\size{\mathcal U_j\setminus\mathcal B_{j+1}}}{\size{\mathcal U_{j}}}
\stackrel{(b)}{<} 0 , \numberthis
\end{align*}
where (a) follows from
$\size{\mathcal U_j\setminus\mathcal B_{j+1}}=
\size{\mathcal U_j}-\size{\mathcal B_{j+1}}$;
 (b) follows 
from $\size{\mathcal B_{j+1}}>0$ and $\mathcal B_{j+1}\subset\mathcal U_j$.  

$\underline{\size{\mathcal B_{j+1}} = 0}:$
In this case, we design $\tilde\mu$ equal to $\mu$ except for the following:
$\tilde\mu_j(\mathcal U_j, D_j) = (0,\tilde{\mathcal B}_j^\prime,R_j/2)$, with $\tilde{\mathcal B}_j^{\prime}$ satisfying the conditions of Theorem \ref{th:opt_comm_beams}, so that state $(\mathcal U_j, D_j)$
transitions 
 to state $(\mathcal U_j, D_j-T R_j/2)$. Moreover $\tilde\mu_{j+1}(\mathcal U_j, D_j-T R_j/2) = (0,\tilde{\mathcal B}_j^{\prime\prime},R_j/2)$, with $\tilde{\mathcal B}_j^{\prime\prime}$ satisfying the conditions of Theorem \ref{th:opt_comm_beams}, so that the system moves to state
 $(\mathcal U_j, D_j-T R_j)$ in slot $j+2$.
Under this new policy, the BS performs data communication in both slots, with rate $R_j/2$.
  Thus, the cost-to-go function under $\tilde \mu$ in slot $j$ is given as
\begin{align*}
\label{eq:ctg_hat_mu}
&V_j^{\tilde\mu}({\mathcal U}_j ,D_j)
=
 \phi_d\left(\frac{R_j}{2},\epsilon\right )\size{{\mathcal U}_j }+  V_{j+1}^{\hat \mu}\left({\mathcal U}_j,D_j-T \frac{R_j}{2}\right)
 \\&
 \qquad\quad= 2 \phi_d\left(\frac{R_j}{2},\epsilon\right )\size{{\mathcal U}_j }+  V_{j+2}^{\hat \mu}\left({\mathcal U}_j,D_j-T R_j\right). \numberthis
\end{align*}
By comparing \eqref{eq:ctg_hat_mu} and \eqref{eq:ctg_j_mu} and using the fact that $\mu$ and $\tilde \mu$ are identical for $k \geq j+2$,  we get 
\begin{align*}
\label{eq:comp2}
V_j^{\tilde\mu}({\mathcal U}_j ,D_j) &- V_j^{\mu}({\mathcal U}_j ,D_j)
\stackrel{(a)}{=} \left[ 
2 \phi_d\left(\frac{R_j}{2},\epsilon\right ) - \phi_d\left({R_j},\epsilon\right )
\right] \size{{\mathcal U}_j }
\stackrel{(b)}{<} 0,\numberthis
\end{align*}
where (a) follows from $\size{\mathcal B_{j+1}}{=}0$; (b) follows 
from the strict convexity of 
 $\phi_d\left({R},\epsilon\right )$ over $R{>}0$, implying that  $2 \phi_d\left(\frac{R_j}{2},\epsilon\right ){<}\phi_d\left(R_j,\epsilon\right )$. 
\eqref{eq:comp1} and \eqref{eq:comp2} imply that $\mu$ does not satisfy Bellman's optimality equation, hence it is suboptimal,
yielding a contradiction.
 The theorem is proved.
\end{proof}

From Theorem \ref{th:sens_comm}, we infer that: 
\begin{corollary}
\label{cl1}
Under an optimal policy $\mu^*$, the frame can be split into 
a beam-alignment phase, followed by a data communication phase until the end of the frame. The duration 
$L^*{\in}\mathcal I$ of beam-alignment is, possibly, a random variable, function of the realization of the beam-alignment process.
\end{corollary}

To capture this phase transition, we introduce the state variable 
$\nabla{\in}\{\mathrm{BA},\mathrm{DC}\}$, denoting that the system is operating in the beam-alignment phase
($\nabla{=}\mathrm{BA}$) or switched to data communication ($\nabla{=}\mathrm{DC}$).
The extended state is denoted as
$\left(\mathcal U_k, D_k,\nabla_k\right)$, with the following DP updates.
If $\nabla_k{=}\mathrm{DC}$, then the system remains in the data communication phase until the end of the frame,
and $\nabla_j{=}\mathrm{DC},\forall j\geq k$, yielding
\begin{align*}
\hat V_k^*({\mathcal U}_k,D_k,\mathrm{DC}) =& 
\min_{0<R\leq D_k/T} \Bigl\{
\phi_d(R,\epsilon)\size{{\mathcal U}_k}
+
\hat V_{k+1}^{*}( {\mathcal U}_{k+1},D_{k}-TR,\mathrm{DC})\Bigr\}. \numberthis
\end{align*}
Using the convexity of $\phi_d(R,\epsilon)$ with respect to $R$, 
it is straightforward to prove the following.
\begin{lemma}
$\hat V_k^*({\mathcal U}_k,D_k,\mathrm{DC})= 
(N-k)\phi_d\left(\frac{D_k}{T(N-k)},\epsilon\right)\size{{\mathcal U}_k}.
$
\end{lemma}
That is, it is optimal to transmit with constant rate $\frac{D_k}{T(N-k)}$
in the remaining $(N-k)$ slots until the end of the frame. On the other hand, if $\nabla_k=\mathrm{BA}$,
then $\nabla_j=\mathrm{BA},\forall j\leq k$ and $D_k=D_0$, since no data has been transmitted yet.
 Then,
\begin{align}
\nonumber
\label{bgsdfg}
&\hat V_k^*({\mathcal U}_k, D_0,\mathrm{BA}) = 
\min\Bigr\{
(N-k)\phi_d\left(\frac{NR_{\min}}{N-k},\epsilon\right)\size{{\mathcal U}_k }
,
\\&
\min_{\mathcal B_k\subset\mathcal U_k }
\phi_s \size{\mathcal B_k } 
+\frac{\size{\mathcal B_k }}{\size{\mathcal U_k }}
\hat V_{k+1}^*({\mathcal B}_k,D_0,\mathrm{BA})
\\&
+\left(1-\frac{\size{\mathcal B_k }}{\size{\mathcal U_k }}\right)
\hat V_{k+1}^*({\mathcal U}_k \setminus{\mathcal B}_k, D_0,\mathrm{BA})
\Bigr\},
\end{align}
where 
the outer minimization reflects an optimization over the actions "switch to data communication in slot $k$ with rate $R_k=\frac{NR_{\min}}{N-k}$," or "perform beam-alignment." The inner minimization represents an optimization over the 2D beam 
 $\mathcal B_k $ used for beam-alignment.
\vspace{-2mm}
 \subsection{Optimality of deterministic beam-alignment duration with fractional-search method}
 \label{s3}
It is important to observe that the proposed protocol is \emph{interactive}, so that the duration of the beam-alignment phase, 
$L^*\in\mathcal I$, is possibly a random variable, function of the realization of the beam-alignment process.
For example, if it occurs that the AoD/AoA is identified with high accuracy, the BS may decide to switch to data communication to achieve energy-efficient transmissions until the end of the frame.
Although it may seem intuitive that $L^*$ should indeed be random, in this section
we will show that, instead, $L^*$ is \emph{deterministic}. 
Additionally, we prove the optimality of a \emph{fractional search method}, 
which dictates the optimal beam design.

To unveil these structural properties, we define $v_k^*({\mathcal U}_k) \triangleq \frac{\hat V_k^*({\mathcal U}_k,D_0,\mathrm{BA})}{\size{\mathcal U_k}}$. Then, 
\eqref{bgsdfg} yields
\begin{align*}
\numberthis
&v_k^*({\mathcal U}_k)
= 
\min\Bigr\{
(N-k)\phi_d\left(\frac{NR_{\min}}{N-k},\epsilon\right)
,
\min_{\rho\in[0,1)}
\phi_s\rho
+\rho^2
v_{k+1}^*({\mathcal B}_{\mathrm t,k})
+(1-\rho)^2
v_{k+1}^*({\mathcal U}_k\setminus{\mathcal B}_k)
\Bigr\},
\end{align*}
where $v_N^{*}(  {\mathcal U}_N){=}\infty$
and we used $\rho$ in place of $\frac{\size{\mathcal B_k}}{\size{\mathcal U_k}}$, with $\rho{<}1$ since $\mathcal B_k{\subset}\mathcal U_k$. Using this fact, we find that
$v_{N-1}^*({\mathcal U}_{N-1})
{=}\phi_d\left(NR_{\min},\epsilon\right)$ is \emph{independent} of 
$\mathcal U_{N-1}$.
By induction on $k$, it is then straightforward to see that 
$v_k^*({\mathcal U}_k)$
is \emph{independent} of $\mathcal U_k,\forall k$.
We thus let $v_k^*{\triangleq}v_k^*({\mathcal U}_k),\forall \mathcal U_k $ to capture this independence,
which is then defined recursively as
\begin{align*}
\label{eq:vk_rec}
v_k^*
= 
\min\Bigl\{&
(N-k)\phi_d\left(\frac{NR_{\min}}{N-k},\epsilon\right)
,
\min_{\rho\in[0,1)}
\phi_s\rho
+\left[\rho^2+\left(1-\rho\right)^2\right]
v_{k+1}^*
\Bigr\}.\numberthis
\end{align*}
The value of $\rho$ achieving the minimum in \eqref{eq:vk_rec} is
$\rho_k=\frac{\size{\mathcal B_k}}{\size{\mathcal U_k}}=\frac{1}{2}\left(1-\frac{\phi_s}{2v_{k+1}^*}\right)^+$,
 yielding
\begin{align*}
&v_k^*
{=} 
\min\Biggr\{\!
\underbrace{(N{-}k)\phi_d\left(\frac{NR_{\min}}{N{-}k},\epsilon\right)}_{\Gamma_k\text{ (data communication)}},
\underbrace{v_{k+1}^*{-}\frac{[(2v_{k+1}^*{-}\phi_s)^+]^2}{8v_{k+1}^*}}_{\Lambda_k\text{ (beam-alignment)}}
\Biggr\}.
\end{align*}

From this decomposition, we infer important properties:
\begin{enumerate}[leftmargin=.60cm]
\item Given $v_k^*$, 
the original value function is obtained as $\hat V_k^*({\mathcal U}_k, D_0,\mathrm{BA})
=v_k^* \size{{\mathcal U}_k}$.
If, at time $k$, $\Gamma_k<\Lambda_k$, then it is optimal to switch to data communication in the remaining $N-k$ slots, with  constant
rate $\frac{NR_{\min}}{N-k}$. 
\item 
Otherwise, it is optimal to perform beam-alignment,
 with beam
$\mathcal B_k \subset \mathcal U_k, \; \; \size{ \mathcal B_k} = \rho_k \size{\mathcal U_k}$.
\item Finally, since the time to switch to data communication is solely based on $\{v_k^*\}$, but not on $\mathcal U_k $, it follows that
\emph{fixed-length} beam-alignment is optimal, with duration
\begin{align}
\label{eq:opt_len}
L^*=\min\left\{k:\Gamma_k<\Lambda_k\right\}.
\end{align}
\end{enumerate}
These structural results are detailed in the following theorem.
\begin{theorem}
\label{optimalpoliuniform}
Let
\begin{align}
\label{Lmin}
\!L_{\min}{=}\!\!\!\!\!\argmin_{L\in\{0,\dots,N{-}1\}}\!\left\{L:(N{-}L)\phi_d\left(\frac{NR_{\min}}{N-L},\epsilon\right){>}\frac{\phi_s}{2}\right\}
\end{align}
and, for $L_{\min}\leq L<N$,
\begin{align}
\label{dfh}
\left\{\begin{array}{l}
v_L^{(L)}
= 
(N-L)\phi_d\left(\frac{NR_{\min}}{N-L},\epsilon\right),
\\
v_k^{(L)}
= 
v_{k+1}^{(L)}-\frac{(2v_{k+1}^{(L)}-\phi_s)^2}{8v_{k+1}^{(L)}},\ k<L.
\end{array}\right.
\end{align}
Then, the beam-alignment phase has \emph{deterministic} duration
\begin{align}
\label{eq:opt_ba_len}
L^*= 
\arg\!\!\!\!\!\!\!\!\min_{\!\!\!\!\!\!L\in \{0\}\cup\{L_{\min},\ldots,N-1\}}
v_0^{(L)}.
 \end{align}
For $0{\leq}k{<}L^*$ (beam-alignment phase),
$\mathcal B_k$ is optimal iff
\begin{align}
\label{hgdfgh2}
\mathcal B_k \subset \mathcal U_k, \; \; \size{\mathcal B_k} = \rho_k\size{\mathcal U_k},
\end{align}
where $\rho_k$ is the \emph{fractional search} parameter, defined as
\begin{align}
\label{fk}
\begin{cases}
&\rho_{L^*-1}=\frac{1}{2}-\frac{\phi_s}{4(N-L^*)\phi_d\left(\frac{NR_{\min}}{N-L^*},\epsilon\right)},\\
&\rho_{k}=\frac{1-\rho_{k+1}}{1-2\rho_{k+1}^2}\rho_{k+1},\ k<L^*-1.
\end{cases}
\end{align}
Moreover, $\rho_k\in(0,1/2)$, strictly increasing in $k$.
For $k\geq L^*$, the data communication phase occurs with rate $\frac{NR_{\min}}{N-L^*}$, and 2D beam given by Theorem \ref{th:opt_comm_beams}.
\end{theorem}
\begin{proof}
Since the optimal duration of the beam-alignment phase is deterministic, as previously discussed, 
 we consider a fixed beam-alignment duration $L$, and then optimize over $L$ to achieve minimum energy consumption. Let $L\in\mathcal I$. Then, the DP updates are obtained by adapting \eqref{eq:vk_rec} to this case (so that the outer minimization disappears for $k<L$), yielding
\begin{align}
\label{vkupdate}
\begin{cases}
&v_L^{(L)}= (N-L)\phi_d\left(\frac{NR_{\min}}{N-L},\epsilon\right),\\
&v_k^{(L)}=g_k(\rho_k),\ k<L,\text{ where}\\
&g_k(\rho)\triangleq\phi_s\rho+\left[\rho^2+\left(1-\rho\right)^2\right]v_{k+1}^{(L)},\\
&\rho_k=\arg\min_{\rho\in[0,1]}g_k(\rho)=\frac{1}{2}\left(1-\frac{\phi_s}{2v_{k+1}^{(L)}}\right)^+.
\end{cases}
\end{align}
Since the goal is to minimize the energy consumption, the optimal $L$ is obtained as $L^*{=}\arg\min_Lv_0^{(L)}$.
We now prove that $0<L<L_{\min}$ is suboptimal, so that this optimization can be restricted to 
$L\in\{0\}\cup\{L_{\min},\cdots,N-1\}$, as in \eqref{eq:opt_ba_len}. Let $0<L<L_{\min}$, so that $v_L^{(L)}\leq\phi_s/2$,
as can be seen from the definition of $L_{\min}$ in \eqref{Lmin}.
 Note that $v_k^{(L)}$ is a non-decreasing function of $k$. In fact, $v_k^{(L)}\leq g_k(0)=v_{k+1}^{(L)}$.
 Then, it follows that $v_k^{(L)}\leq \phi_s/2,\forall k$,
 hence $\rho_k=0,\forall k$, yielding
 $v_0^{(L)}=v_L^{(L)}$ by induction. However, $v_L^{(L)}$ is an increasing function of $L$ (it is more energy efficient to spread transmissions over a longer interval), hence
 $v_0^{(L)}>v_0^{(0)}$ and such $L$ is suboptimal. This proves that any $0<L<L_{\min}$ is suboptimal.

We now prove the updates for $L\geq L_{\min}$, i.e., $v_L^{(L)}>\phi_s/2$. 
 By induction, we have that $v_k^{(L)}>\phi_s/2,\forall k$. 
 In fact, this condition trivially holds for $k=L$, by hypothesis.
Now, assume $v_{k+1}^{(L)}>\phi_s/2$ for some $k<L$. Then,
 $v_k^{(L)}=\min_{\rho\in[0,1]}g_k(\rho)$, minimized at 
 $\rho_k=\frac{1}{2}\left(1-\frac{\phi_s}{2v_{k+1}^{(L)}}\right)$, so that 
 $v_k^{(L)}=g_k(\rho_k)$, yielding \eqref{dfh}.
 This recursion is an increasing function of $v_{k+1}^{(L)}$, yielding $v_k^{(L)}>\phi_s/2$,
 thus proving the induction. It follows that 
 $\rho_k=\frac{1}{2}\left(1-\frac{\phi_s}{2v_{k+1}^{(L)}}\right),\forall k$, yielding the 
 recursion given by \eqref{dfh}.
 The fractional search parameter $\rho_k$ is finally obtained by substituting $v_{k+1}^{(L)}=\frac{\phi_s}{2(1-2\rho_k)}$ into the recursion \eqref{dfh} to find a recursive expression of
 $\rho_{k}$ from $\rho_{k+1}$, yielding \eqref{fk}.
 These fractional values are used to obtain $\mathcal B_k$ in \eqref{hgdfgh2}.
 
 To conclude, we show by induction that $\rho_k{\in}(0,1/2)$, strictly increasing in $k$. 
 This is true for $k{=}L{-}1$ since
 $\rho_{L-1}{\in}(0,1/2)$. Assume that $\rho_{k+1}{\in}(0,1/2)$, for some $k{\leq}L{-}2$. Then,
by inspection of \eqref{fk}, it follows that $0{<}\rho_k{<}\rho_{k+1}<1/2$.
The theorem is thus proved.
\end{proof}
\vspace{-5mm}
\section{Decoupled BS and UE Beam-Alignment}
\label{sec:rec_beams_optimalality}
In the previous section, we proved the optimality of a fractional search method, based on an extended action space that uses the 2D beam $\mathcal B_k{\in}[-\pi,\pi]^2$, which may take any shape.
However, actual beams should satisfy the rectangular constraint $\mathcal B_k{=}\mathcal B_{\mathrm t, k}{\times}\mathcal B_{\mathrm r, k}$, and therefore, it is not immediate to see that the optimal scheme 
outlined in Theorem \ref{optimalpoliuniform} is attainable in practice.
Indeed, in this section we prove that there exists a feasible beam design attaining optimality.
The proposed beam design decouples over time the beam-alignment of the AoD at the BS (\emph{BS beam-alignment}) and of the AoA at the UE (\emph{UE beam-alignment}). To explain this approach, we define the support of the marginal belief with respect to $\theta_{\mathrm x},\mathrm x{\in}\{ \mathrm t, \mathrm r\}$ as $\mathcal U_{\mathrm x, k} \equiv \supp( f_{\mathrm x, k})$. 
In \emph{BS beam-alignment}, indicated with $\beta_k{=}1$, the 2D beam is chosen as
 $\mathcal B_k{=}\mathcal B_{\mathrm t, k}{\times}\mathcal U_{\mathrm r, k}$, where  $\mathcal B_{\mathrm t, k}{\subset}\mathcal U_{\mathrm t, k} $, so that the BS can better estimate the support of the AoD, whereas the UE receives over the entire support of the AoA.
On the other hand, in \emph{UE beam-alignment}, indicated with $\beta_k{=}2$, the 2D beam is chosen as $\mathcal B_k{=}\mathcal U_{\mathrm t, k}{\times}\mathcal B_{\mathrm r, k} $, where  $\mathcal B_{\mathrm r, k}{\subset}\mathcal U_{\mathrm r, k} $, so that the UE can better estimate the support of the AoA, whereas the BS transmits over the entire support of the AoD.
 We now define a policy $\mu$ that uses this principle, and then prove its optimality.
\begin{defi}[Decoupled fractional search policy]
\label{fracpol}
Let $L^*$, $\vartheta$, $\{\rho_k{:}k{=}0,\dots,L^*{-}1\}$ as in Theorems~\ref{th:opt_comm_beams},~\ref{optimalpoliuniform}.
In slots $k{=}L^*,\dots,N$, data communication occurs with rate 
$R_k{=}\frac{NR_{\min}}{N-L^*}$ and beams
\begin{align}
\label{eq:optimal_comm_beams1}
\mathcal B_{\mathrm t, k} \subseteq \mathcal U_{\mathrm t, k},\  
\mathcal B_{\mathrm r, k} \subseteq \mathcal U_{\mathrm r, k},
\; \; 
&\size{\mathcal B_{\mathrm t, k}}\size{\mathcal B_{\mathrm r, k}} = \vartheta\size{\mathcal U_{\mathrm t, k}}
\size{\mathcal U_{\mathrm r, k}}.
\end{align}
In slots $k{=}0,1,\dots,L^*$,  $\beta_k{\in}\{1,2\}$ is chosen arbitrarily and beam-alignment occurs with beams
\begin{align}
\!\begin{cases}
\mathcal B_{\mathrm t, k} \subset \mathcal U_{\mathrm t, k},
\ \mathcal B_{\mathrm r, k}=\mathcal U_{\mathrm r, k},
 \; \; \size{\mathcal B_{\mathrm t, k}} = \rho_k\size{\mathcal U_{\mathrm t, k}}, \text{if }\beta_k{=}1
\\
\mathcal B_{\mathrm t, k} = \mathcal U_{\mathrm t, k},
\ \mathcal B_{\mathrm r, k}\subset\mathcal U_{\mathrm r, k},
 \; \; \size{\mathcal B_{\mathrm r, k}} = \rho_k\size{\mathcal U_{\mathrm r, k}},
\text{if }\beta_k{=}2.
\end{cases}\!\!\!\!\!\!
\end{align}
\end{defi}
\begin{theorem}
\label{thm:decoupled}
The decoupled fractional search policy is optimal, with minimum power consumption
\begin{align}
\label{Powunif}
\bar P_{\mathrm{u}}=\frac{v_0^{(L^*)}}{T_{\mathrm{fr}}}\size{\mathcal U_0}.
\end{align}
\end{theorem}
\begin{proof}
The proof is provided in Appendix \ref{app:proof_of_th_perf_with_errors}.
\end{proof}
{The intuition behind this result is that, by decoupling the beam-alignment of the AoD and AoA over time,
the proposed method maintains a rectangular support
$\mathcal U_k = \mathcal U_{\mathrm t,k} \times \mathcal U_{\mathrm r,k}$, so that 
no loss of optimality is incurred by using
a rectangular beam 
$\mathcal B_k = \mathcal B_{\mathrm t,k} \times \mathcal B_{\mathrm r,k}$.} Additionally, we can infer that the \emph{exhaustive search} method is suboptimal, since
it searches over the AoD/AoA space in an exhaustive manner, rather than by decoupling this search over time.
\section{Non-Uniform prior}
\label{sec:nonunifprior}
In this section,
we investigate the case of non-uniform prior $f_0$.
We use the previous analysis to design a heuristic scheme with performance guarantees.
We consider the decoupled fractional search policy (Definition \ref{fracpol}), with the following additional constraints:
in the beam-alignment phase $k<L^*$,
if $\beta_k^* = 1$ (BS beam-alignment), then
\begin{align}
\label{v1}
\!\!\mathcal B_{\mathrm t,k}^*{=}\arg\max_{\mathcal B_{\mathrm t,k}\subset\mathcal U_{\mathrm t,k}} \int_{\mathcal B_{\mathrm t, k}} \!\!\!\!\!\!\! f_{\mathrm t, k}(\theta_{\mathrm t}) \mathrm d \theta_{\mathrm t},\ \mathrm{s.t.}\ \size{\mathcal B_{\mathrm t,k}} = \rho_k \size{\mathcal U_{\mathrm t,k}};
\end{align}
 if $\beta_k^* = 2$ (UE beam-alignment), then 
 \begin{align}
 \label{v2}
\!\!\mathcal B_{\mathrm r,k}^*{=}\arg\max_{\mathcal B_{\mathrm r, k}\subset\mathcal U_{\mathrm r,k}} \int_{\mathcal B_{\mathrm r, k}}\!\!\!\!\!\!\! f_{\mathrm r, k} (\theta_{\mathrm r})\mathrm d \theta_{\mathrm r},\ \mathrm{s.t.}\ \size{\mathcal B_{\mathrm r,k}} = \rho_k \size{\mathcal U_{\mathrm r,k}}.
\end{align}
Hence, the probability of ACK  can be bounded as
\begin{align}
\label{eq:ack_nonunif}
\left.
\begin{array}{lr}
\text{Case }\beta_k^* = 1\text{:}&
\int_{\mathcal B_{\mathrm t, k}^*}  f_{\mathrm t, k}(\theta_{\mathrm t}) \mathrm d \theta_{\mathrm t} \geq \frac{\size{\mathcal B_{\mathrm t,k}^*}}{\size{\mathcal U_{\mathrm t,k}}}\\
\text{Case }\beta_k^* = 2\text{:}&
\int_{\mathcal B_{\mathrm r, k}^*}  f_{\mathrm r, k}(\theta_{\mathrm r}) \mathrm d \theta_{\mathrm r} \geq \frac{\size{\mathcal B_{\mathrm r,k}^*}}{\size{\mathcal U_{\mathrm r,k}}}
\end{array}
\right\}=\rho_k.
\end{align}
In other words, such choice of the BS-UE beam 
maximizes the probability of successful beam-detection, so that the resulting probability of ACK is at least as good as in the uniform  case.

Similarly, in the data communication phase $k\geq L^*$, the
BS transmits with rate $R_k=\frac{NR_{\min}}{N-L^*}$, and the
 beams are chosen as in Definition \ref{fracpol},
with the additional constraint
\begin{align*}
(\mathcal B_{\mathrm t,k}^*,\mathcal B_{\mathrm r,k}^*) =& \arg\max_{\mathcal B_{\mathrm t,k}\times\mathcal B_{\mathrm r,k}\subseteq\mathcal U_k} \int_{\mathcal B_{\mathrm t, k}\times\mathcal B_{\mathrm r, k}}  f_{k}(\bm\theta) \mathrm d \bm\theta,\;\; \mathrm{s.t.}\; \; 
{\size{\mathcal B_{\mathrm t,k}}\size{\mathcal B_{\mathrm t,k}}} ={\vartheta} {\size{{\mathcal U}_{\mathrm t,k}}\size{{\mathcal U}_{\mathrm r,k}}}. \numberthis
\end{align*}
Under this choice, the energy consumption per data communication slot is obtained from \eqref{eq:cost1},
\begin{align}
\label{Ekineq}
& 
E_k  = 
{\psi}_d(R_k)
\frac{\size{\mathcal B_k}}{\bar{F}_{\gamma}^{-1}\left(\frac{1-\epsilon}{\mathbb{P}(\bm{\theta} \in \mathcal B_k | \mathcal U_k)} \right)}
\\&
\overset{(a)}{\leq}
{\psi}_d(R_k)
\frac{\size{\mathcal B_k}}{\bar{F}_{\gamma}^{-1}\left(\frac{(1-\epsilon)\size{\mathcal U_k}}{\size{\mathcal B_k}} \right)}
\overset{(b)}{=}
\phi_d(R_k,\epsilon)\size{\mathcal U_k},
\end{align}
where (a) follows from 
$\mathbb{P}(\bm{\theta}{\in}\mathcal B_k | \mathcal U_k){\geq}|\mathcal B_k|/|\mathcal U_k|$,
and (b) from
$\size{\mathcal B_{\mathrm t, k}}\size{\mathcal B_{\mathrm r, k}} = \vartheta \size{\mathcal U_{\mathrm t, k}}
\size{\mathcal U_{\mathrm r, k}}$, and from \eqref{phid} with $\vartheta = (1-\epsilon)/q^*$
(Theorem \ref{th:opt_comm_beams}). This result implies that
data communication is more energy efficient than in the uniform case, see \eqref{eq:comm_val_ftn_final}. {These observations suggest that the uniform prior yields the worst performance, 
as confirmed by the following theorem.}
\begin{theorem}
The minimum power consumption for the non-uniform prior is upper bounded by
$\bar P_{\mathrm{nu}}\leq \bar P_{\mathrm{u}}$,
with equality when $f_0$ is uniform. 
\end{theorem}
\begin{proof}
We denote the value function of the non-uniform case under such policy as $V_{\mathrm{nu}, k}(\mathcal U_{k},D_k)$.
Additionally, we let $\bar P_{\mathrm{nu}}$
be the corresponding minimum power consumption, solution of problem $ \mathrm{P}_2$ in \eqref{P2}, to distinguish it from the minimum power consumption in the uniform case,
given by \eqref{Powunif}.
 For $k=L^*$ (data communication begins),
\eqref{Ekineq} implies that
\begin{align}
\label{eq:ctgo_nonuniform1}
&V_{\mathrm{nu}, k}(\mathcal U_{k},D_0)\leq
(N-L^*)\phi_d\left(\frac{NR_{\min}}{N-L^*},\epsilon\right)
\size{\mathcal U_{k}}.
\end{align}
For $k <L^*$ (beam-alignment phase), 
it can be expressed as
\begin{align*}
\label{eq:ctgo_nonuniform2}
&\!V_{\mathrm{nu}, k}(\mathcal U_{k},D_0)
{=}\phi_s\size{\mathcal B_{k}^*}{+}\!
\int_{\mathcal B_{k}^*}\!\!\!\!\! f_{k}(\bm\theta) \mathrm d \bm\theta
V_{\mathrm{nu}, k+1}(\mathcal B_{k}^*,D_0) +
\left(1{-}\int_{\mathcal B_{k}^*} \! \!\!\! f_{k}(\bm\theta) \mathrm d \bm\theta\!\!\right)
\!V_{\mathrm{nu}, k+1}(\mathcal U_k\setminus\mathcal B_{k}^*,D_0),\numberthis
\end{align*} 
where $\mathcal B_{k}^*$ is given by \eqref{v1} or \eqref{v2}.
The minimum power consumption is given by
$\bar P_{\mathrm{nu}}
{=}V_{\mathrm{nu}, 0}(\mathcal U_{0},D_{0}){/}T_{\mathrm{fr}}$, 
so that $\bar P_{\mathrm{nu}}\leq \bar P_{\mathrm{u}}$ is equivalent to $V_{\mathrm{nu},k}(\mathcal U_{k},D_{k}){\leq}
v_k^{(L^*)}\size{\mathcal U_{k}}$ when $k{=}0$.
We prove this inequality for general $k$ by induction.
The induction hypothesis holds for $k{=}L^*$, see
\eqref{eq:ctgo_nonuniform1} with
$v_{L^*}^{(L^*)}$ given in \eqref{dfh}.
Assume it holds for $k+1$, where $k\leq L^*-1$.
Then,
  (\ref{eq:ctgo_nonuniform2}) can be expressed as
\begin{align*}
& V_{\mathrm{nu}, k}(\mathcal U_{k},D_{0})
\leq
\phi_s\size{\mathcal B_{k}^*}
{+}\int_{\mathcal B_{k}^*} \!\!\!f_{k}(\bm\theta) \mathrm d \bm\theta 
v_{k+1}^{(L^*)}\size{\mathcal B_{k}^*}
{+}\left(1{-}\int_{\mathcal B_{k}^*}  f_{k}(\bm\theta) \mathrm d \bm\theta \right) v_{k+1}^{(L^*)}\size{\mathcal U_{k} \setminus \mathcal B_{k}^*}
\\&
\stackrel{(a)}{=}
\left[\phi_s\rho_k+v_{k+1}^{(L^*)}\left(1-2\rho_k+2\rho_k^2\right)\right]\size{\mathcal U_{k}}
 -\left(\int_{\mathcal B_{k}^*}  f_{k}(\bm\theta) \mathrm d \bm\theta-\rho_k\right)v_{k+1}^{(L^*)}\size{\mathcal U_{k}}\left(1-2\rho_k\right),
\end{align*}
where (a) follows from \eqref{v1}-\eqref{v2} and $\size{\mathcal U_{k} \setminus \mathcal B_{k}^*}
=\size{\mathcal U_{k}}-\size{\mathcal B_{k}^*}$.
Finally, the bound \eqref{eq:ack_nonunif} yields
\begin{align*}
& V_{\mathrm{nu}, k}(\mathcal U_{k},D_{0})
\leq
\left[\phi_s\rho_k+v_{k+1}^{(L^*)}\left(1-2\rho_k+2\rho_k^2\right)\right]\size{\mathcal U_{k}}
=
v_k^{(L^*)}\size{\mathcal U_{k}},
\end{align*}
where the last equality is obtained by using the recursion \eqref{dfh} and the fact that 
$\rho_k=\frac{1}{2}-\frac{\phi_s}{4v_{k+1}^{(L^*)}}$ (see proof of Theorem \ref{optimalpoliuniform}). This proves the induction step.
 Clearly, equality is attained in the uniform case.
The theorem is thus proved.
\end{proof}

This result is in line with the fact that one can leverage the structure of the joint distribution over $\bm\theta$ to improve the beam-alignment algorithm. However, for the first time to the best of our knowledge, this result provides a heuristic scheme with provable performance guarantees.
\section{Impact of False-alarm and Misdetection}
\label{sec:fa_md_impact}
In this section, we analyze the impact of false-alarm and misdetection on the performance of the decoupled fractional search policy (Definition \ref{fracpol}). For simplicity, we focus only on the uniform prior case.
Under false-alarm and misdetection, the MDP introduced in  Sec. \ref{sec:probform} does not follow the Markov property. To overcome this problem, we augment it with the state variable $e_k\in\{0,1\}$,
with $e_k=0$ iff no errors have been introduced up to slot $k$.
 Note that, if errors have been introduced ($e_k=1$), then necessarily
$\bm\theta\notin\mathcal U_k$, so that we can write $e_k=1-\chi(\bm\theta\in\mathcal U_k)$.
  It should be noted that $e_{k}$ is not observable in reality and is considered for the purpose of analysis only (indeed, the policy under analysis does not use such information). 
We thus define the state as $(\mathcal U_k,e_k)$,\footnote{The backlog $D_k$ is removed from the state space, since no data is transmitted during the beam-alignment phase.} and study the transition probabilities during the beam-alignment phase $k<L^*$. From state $(\mathcal U_k,0)$ (no errors have been introduced), the transitions are
\begin{align*}
\label{eq:trs_1}
\!\!\!(\mathcal U_{k+1},e_{k+1})
\!{=}\!\left\{\begin{array}{ll}
\!\!\!(\mathcal B_{k},0), &\!\!\text{w.p. }\rho_k(1-p_{\mathrm{md}})\\
\!\!\!(\mathcal B_{k},1), &\!\!\text{w.p. }\left(1-\rho_k\right)p_{\mathrm{fa}}\\
\!\!\!(\mathcal U_k\setminus\mathcal B_{k},0), &\!\!\text{w.p. }\left(1-\rho_k\right)(1-p_{\mathrm{fa}})
\\
\!\!\!(\mathcal U_k\setminus\mathcal B_{k},1), &\!\!\text{w.p. }\rho_k\ p_{\mathrm{md}},
\end{array}
\right.\!
\numberthis
\end{align*}
where $p_{\mathrm{fa}}$ and $p_{\mathrm{md}}$ denote the false-alarm and misdetection probabilities, respectively.
In fact, if no errors occur, then 
$\bm{\theta}{\in}\mathcal B_k$ with probability
$\frac{\size{\mathcal B_{k}}}{\size{\mathcal U_{k}}}{=}\rho_k$
and $\bm\theta{\notin}\mathcal B_k$ otherwise, yielding the first and third cases;
if a false-alarm or misdetection error is introduced, then 
the BS infers incorrectly that $\bm\theta{\in}\mathcal B_k$ (second case)
or $\bm\theta{\notin}\mathcal B_k$ (fourth case), respectively,
and the new state becomes $e_{k+1}{=}1$.
Once errors have been introduced (state $(\mathcal U_k,1)$), 
it follows that $\bm\theta{\notin}\mathcal B_k$, so that
$\mathcal U_{k+1}{=}\mathcal B_k$ iff a false-alarm error occurs, and the transitions are
\begin{align*}
\label{eq:trs_2}
(\mathcal U_{k+1},e_{k+1})
=\begin{cases}
(\mathcal B_{k},1), & \text{w.p. }p_{\mathrm{fa}}\\
(\mathcal U_k\setminus\mathcal B_{k},1), & \text{w.p. }1-p_{\mathrm{fa}}.
\end{cases}
\numberthis
\end{align*}

The average throughput and power are given by
\begin{align*}
\label{eq:avg_th_1}
&\bar T_{\mathrm{err}} = \mathbb{E}\left[ (1-e_{L^*})(1-\epsilon)R_{\min}|\mathcal U_0,e_0=0\right],\\
&\bar P_{\mathrm{err}} =
\frac{1}{T_{\mathrm{fr}}}\mathbb E\Biggr[
\phi_s\sum_{k=0}^{L^*-1}\rho_k\size{\mathcal U_k}\\&+
\left.(N-L^*)\phi_d\left(\frac{NR_{\min}}{N-L^*},\epsilon\right)\size{\mathcal U_{L^*}}
\right|\mathcal U_0,e_0=0\Biggr]. \numberthis
\end{align*}
In fact, a rate equal to $R_{\min}$ is sustained if: (1) no outage occurs in the data communication phase, with probability $1-\epsilon$; (2) no errors occur during the beam-alignment phase, $e_{L^*}=0$.

The analysis of the underlying Markov chain $\{(\mathcal U_k,e_k)$, $k{\geq}0\}$
yields the following theorem.
\begin{theorem}
\label{th:perf_with_errors}
Under the decoupled fractional search policy,
\begin{align}
\label{eq:thr_frac}
&\bar T_{\mathrm{err}}
{=}(1{-}\epsilon) R_{\min} \prod_{k=0}^{L^*{-}1}\Bigr[\left(1{-}\rho_k\right)(1{-}p_{\mathrm{fa}}){+}\rho_k(1{-}p_{\mathrm{md}}) \Bigr],
\\
&
\label{eq:pow_frac}
\bar P_{\mathrm{err}}=
\bar P_{\mathrm{u}}+\frac{h_0 + u_0}{T_{\mathrm{fr}}}\size{\mathcal U_0},
\end{align}
where  $\bar P_{\mathrm{u}}$ in
\eqref{Powunif} is the error-free case, and we have defined $h_{L^*}{=}u_{L^*}{=}0$ and, for $k{<}L^*$,
\begin{align}
\label{eq:h_k}
&h_k = 
\phi_s\frac{\rho_k-p_{\mathrm{fa}}}{2}+\left[\rho_kp_{\mathrm{fa}}+\left(1-\rho_k\right)(1-p_{\mathrm{fa}})\right]h_{k+1},\\
\label{eq:u_k}
&u_k = 
\left[\rho_k^2(1{-}p_{\mathrm{md}}){+}\left(1{-}\rho_k\right)^2(1{-}p_{\mathrm{fa}})\right]u_{k+1}
-(1{-}p_{\mathrm{fa}}{-}p_{\mathrm{md}}){\rho_k}
\left[\frac{\phi_s}{2}{+}h_{k+1}\left(1{-}2\rho_k\right)\right].
\end{align}
\end{theorem}
\begin{proof}
The proof is provided in Appendix C.
\end{proof}
\section{Numerical Results}
\label{numres}
In this section, we demonstrate the performance of the proposed
\emph{decoupled fractional search} (DFS) scheme and compare it with the \emph{bisection search} algorithm developed in \cite{new_benchmark} and two variants of \emph{exhaustive search}. 
{In the bisection algorithm \cite{new_benchmark} (BiS), in each beam-alignment slot the uncertainty region is divided into two regions of equal width, scanned in sequence by the BS by transmitting beacons corresponding to each region. 
Then, the UE compares the signal power (the strongest indicating alignment) and transmits the feedback to the BS. 
Since in each beam-alignment slot two sectors are scanned (each of duration $T_B$), the
total duration of the beam-alignment phase is $(2 T_B + T_F)L\;[\mathrm s]$, where $T_F$ is the feedback time.}
In \emph{conventional exhaustive search} (CES), the BS-UE scan exhaustively the entire beam space.
In the BS beam-alignment sub-phase, the
 BS searches over $N_{B}^{(BS)}$ beams covering the AoD space, while the UE receives isotropically;
 in the second UE beam-alignment sub-phase, the 
 BS transmits using the best beam found in the first sub-phase, whereas the UE searches exhaustively over 
 $N_{B}^{(UE)}$ beams covering the AoA space.
{Since the UE reports the best beam at the end of each sub-phase,
 the total duration of the beam-alignment phase is $[N_{B}^{(BS)}+N_{B}^{(UE)}]T_B+2T_F$.}
On the other hand, in the \emph{interactive exhaustive search} (IES) method, 
the UE reports the feedback at the end of each beam-alignment slot,
and each beam-alignment sub-phase terminates upon receiving an $\mathrm{ACK}$ from the UE.
Since the BS awaits for feedback at the end of each beam, the duration of the beam-alignment phase is $(T_B+T_F)
[\hat N_{B}^{(BS)}+\hat N_{B}^{(UE)}]
$, where $\hat N_B\leq N_B$ is the number of beams scanned until receiving an ACK;
{assuming the AoD/AoA is uniformly distributed over the beam space, the expected duration of 
the beam-alignment phase is then 
$\frac{1}{2}(T_B+T_F)\left[N_{B}^{(BS)}+N_{B}^{(UE)}+2\right]$.
}

\par We use the following parameters: [carrier frequency]${=}30\mathrm{GHz}$, $d{=}10 \mathrm{m}$, [path loss exponent]${=}2$, $T_{\mathrm{fr}}{=}20\mathrm{ms}$, $T_B{=}50\mu\mathrm{s}$, $T_F{=}50\mu\mathrm{s}$, $\size{\mathcal U_0}{=}[\pi]^2$, $N_0{=}-173\mathrm{dBm}$, $W_{\mathrm{tot}}{=}500 \mathrm{MHz}$, $M_{\mathrm t}{=}M_{\mathrm r}{=}128$. 

\begin{figure}
\centering
\includegraphics[width=.8\linewidth,trim={15 0 30 10},clip=true]{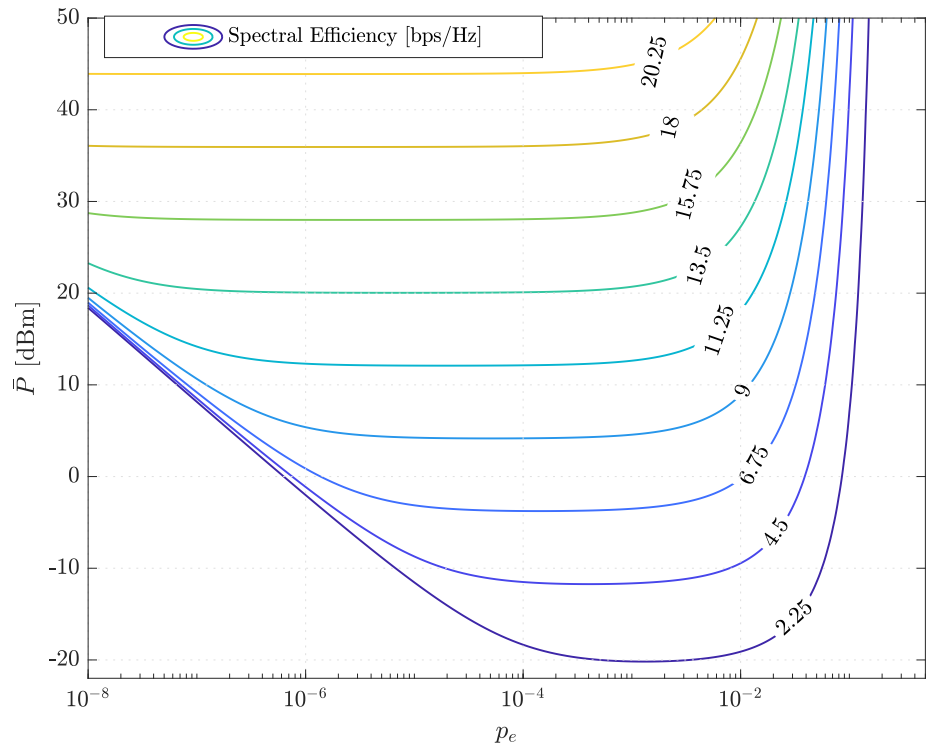}
\caption{Spectral efficiency versus beam-alignment error probability $p_e$ for DFS.}
  \label{fig:thputvspe}
\vspace{-6mm}
\end{figure}

\begin{figure}
\centering
\includegraphics[width=.8\linewidth,trim={15 0 30 10},clip=true]{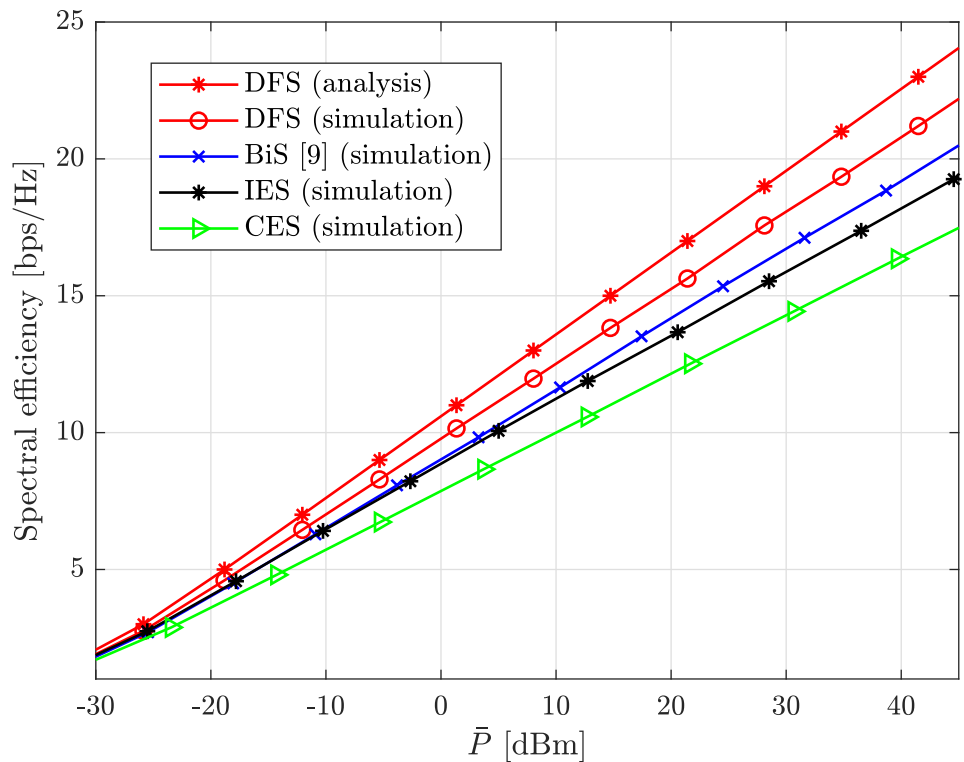}
\caption{Spectral efficiency versus average power consumption.}
  \label{fig:energy_consumption}
\vspace{-6mm}
\end{figure}

\par In Fig.~\ref{fig:thputvspe}, we depict the average power vs the probability of false-alarm and misdetection $p_e$ for different values of the spectral efficiency using expressions \eqref{eq:thr_frac} and \eqref{eq:pow_frac}. We use $\epsilon= 0.01$, and consider Rayleigh fading with no CSI at BS, corresponding to $h{\sim}\mathcal{CN}(0,1/\ell(d))$ with $\hat h{=}0$ and $\sigma_e^2{=}1/\ell(d)$. We restrict the optimization of $L$ over $L{\in}\{0,\ldots,L_{\max}\}$, to capture a maximum resolution constraint for the antenna array, where we chose $L_{\max}=14$.
  From the figure, we observe that, for a given $p_e$, as the spectral efficiency increases so does the average power consumption due to increase in the energy cost of data communication. Moreover, the figure reveals that, for a given value of spectral efficiency, there exists an optimal range of $p_e$, where power consumption is minimized. The performance degrades for $p_e$ above the optimal range due to false-alarm and misdetection errors during beam-alignment,
  causing outage in data communication;
   similarly, it degrades for $p_e$ below the optimal range due to an increased power consumption of beam-alignment.

\par\label{p10}
In Fig.~ \ref{fig:energy_consumption}, we plot the results of a Monte-Carlo simulation with analog beams
 generated using the algorithm in \cite{rate_maximization2}. 
In this case, we obtain $\phi_s{=}-94\mathrm{dBm}$ with $p_{\mathrm{fa}}{=}p_{\mathrm{md}}{=}10^{-5}$. For BiS and DFS we set $L_{\max}=10$ to capture a maximum resolution constraint for the antenna array; for the exhaustive search methods, we choose $N_B^{(BS)}=N_B^{(UE)}{=}32$. The performance gap between the analytical and the simulation-based curves for DFS is attributed to the fact that the beams used in the simulation have non-zero side-lobe gain and non-uniform main-lobe gain, as opposed to the "sectored" beams used in the analytical model. This results in false-alarm, misdetection errors, and leakage, which lead to some performance degradation. 
\emph{However, the simulation is in line with the analytical curve, and exhibits superior performance compared to the other schemes, thus demonstrating that the analysis using the sectored gain model provides useful insights for practical design.}
\begin{figure}
\centering
\includegraphics[width=.8\linewidth,trim={15 0 30 10},clip=true]{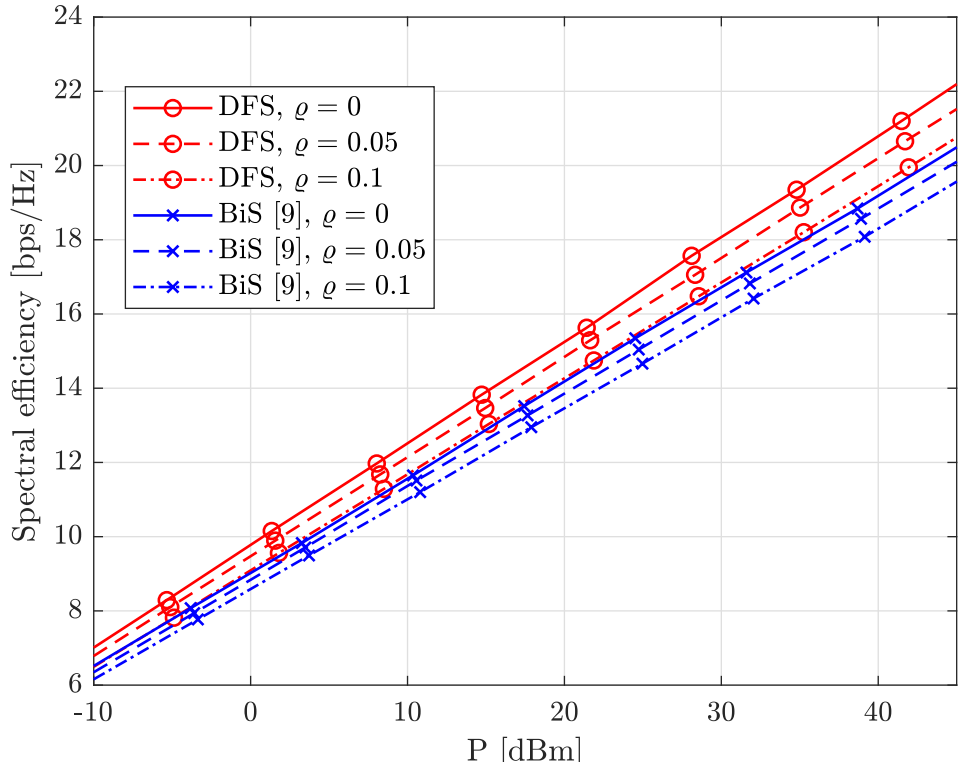}
\caption{Performance degradation with multi-cluster channel ($K=2$).}
\label{fig:mpath}
\vspace{-6mm}
\end{figure}
For instance, to achieve a spectral efficiency of $15 \mathrm{bps}/\mathrm{Hz}$, BiS \cite{new_benchmark} requires $4\mathrm{dB}$ more average power than DFS,
mainly due to the time and energy overhead of scanning two sectors in each beam-alignment slot,
whereas IES and CES require $7.5\mathrm{dB}$ and $14\mathrm{dB}$ more power, respectively.
The performance degradation of IES and CES is due to
the exhaustive search of the best sector, which demands a huge time overhead.
Indeed, IES outperforms CES since it stops beam-alignment once a strong beam is detected.
 \par\label{page11}{
 So far in our analysis, we assumed a channel with a single cluster of rays, see \eqref{eq:channel_matrix}.
 In Fig.~\ref{fig:mpath}, we depict the performance of DFS and BiS \cite{new_benchmark} in a multi-cluster channel ($K=2$ in \eqref{eq:channel_matrixx}), with the weakest cluster having a fraction $\varrho$ of the total energy, $0\leq\varrho\leq 0.1$. It can be seen that the performance of both DFS and BiS degrade as $\varrho$ increases, since a portion of the energy is lost in the weaker cluster, and
 the algorithms may misdetect the weaker cluster instead of the strongest one.
For example, for spectral efficiency of $15\mathrm{bps}/\mathrm{Hz}$, both schemes exhibit $\sim\!2 \mathrm{dB}$ and $\sim\!5 \mathrm{dB}$ performance loss at $\varrho = 5\%$ and $\varrho = 10\%$, respectively, compared to $\varrho=0$ (single cluster). 
 However, DFS consistently outperforms BiS, with a gain of $\sim\!3.5 \mathrm{dB}$. This evaluation demonstrates the robustness of the proposed algorithm in multi-cluster scenarios.}
\section{Conclusions}
\label{sec:concl}
In this paper, we designed an optimal interactive beam-alignment scheme,  with the goal of
minimizing power consumption under a rate constraint.
For the case of perfect detection and uniform prior on AoD/AoA,
 we proved that the optimal beam-alignment protocol has fixed beam-alignment duration,
and that a \emph{decoupled fractional search} method is optimal. Inspired by this scheme, we proposed a heuristic policy for the case of a non-uniform prior, 
and showed that the uniform prior is the worst-case scenario. Furthermore, we
investigated the impact of beam-alignment errors on the average throughput and power consumption. The numerical results depicted the superior performance of our proposed scheme, with up to $4 \mathrm{dB}$, $7.5 \mathrm{dB}$, and $14 \mathrm{dB}$  gain compared to a state-of-the-art bisection search, conventional exhaustive search and interactive exhaustive search policies, respectively, and robustness against multi-cluster channels.
\appendices

\newpage

\section{Proof of Lemma \ref{lemma:sufficient_statistcs}}
\label{proofoflemlemma:sufficient_statistcs}
\begin{proof}
We need the following lemma.
\begin{lemma}
\label{lem1}
Given $f_k,\mathbf a_k,C_k$, the belief $f_{k+1}$ is computed as
\begin{align}
\label{eq:belief_update}
f_{k+1}(\bm \theta)
=
\begin{cases}
\frac{\chi_{\mathcal{B}_k}(\bm \theta )}{\int_{\mathcal{B}_k}f_k(\tilde {\bm \theta})\mathrm d\tilde{\bm\theta} }f_k(\bm \theta), &k\in \mathcal I_s,C_k=\mathrm{ACK},\\
\frac{1-\chi_{\mathcal{B}_k}(\bm \theta )}{1-\int_{\mathcal{B}_k}f_k(\tilde {\bm \theta})\mathrm d\tilde{\bm\theta} }f_k(\bm \theta), &k\in \mathcal I_s,C_k=\mathrm{NACK},\\
f_{k}(\bm \theta), &k\in \mathcal I_d,C_k=\mathrm{NULL}.
\end{cases}
\end{align}
\end{lemma}
\begin{proof}
We denote AoD/AoA random variables pair by $\bm\Theta \triangleq (\Theta_{\mathrm t},\Theta_{\mathrm r} )$ and its realization by $\bm \theta \triangleq (\theta_{\mathrm t},\theta_{\mathrm r} )$.
First note that for $0\leq k\leq N-1$, we have
\begin{align*}
\label{eq:markov_property}
f_{k+1}(\bm \theta) &= f(\bm \Theta  = \bm \theta |\mathbf a^{k},C^{k-1},C_k=c_k)\\
&\stackrel{(a)}{=} \frac{\mathbb P (C_k=c_k|A^{k},C^{k-1}, \bm \Theta =  \bm \theta) f(\bm \Theta = \bm\theta|\mathbf a^{k},C^{k-1})}{\int_{-\pi}^{\pi} \mathbb P (C_k=c_k|A^{k},C^{k-1}, \bm\Theta = \tilde{\bm\theta}) f(\bm\Theta = \tilde{\bm\theta} |\mathbf a^{k},C^{k-1})d \tilde {\bm \theta}}\\
&\stackrel{(b)}{=}\frac{\mathbb P (C_k=c_k|\mathbf a_{k},\bm \Theta =  \bm \theta ) f_k(\bm\theta)}{\int_{-\pi}^{\pi} \mathbb P (C_k=c_k|\mathbf a_{k},\bm \Theta = \tilde {\bm \theta} ) f_k(\tilde{\bm\theta} ) d\tilde{\bm\theta}} \numberthis
\end{align*}
where we have used Bayes' rule in step (a); (b) is obtained by using the fact that, given $\bm \Theta = \bm \theta$, $C_k$ is a deterministic function of $(\mathbf a_k,\bm\theta)$ , independent of $\mathbf a^{k-1},C^{k-1}$; additionally, we used the fact that $f_k(\bm \theta ) =  f(\bm \Theta = \bm \theta |\mathbf a^{k},C^{k-1})$ since $\bm \Theta $ is independent of $\mathbf a_{k}$ given $(\mathbf a_{k-1},C^{k-1})$ .
Now consider the case $k \in \mathcal I_s$, i.e., $\xi_{k}= 1$ and $C_k=\mathrm{ACK}$. Then, we can use (\ref{eq:markov_property}) to get
\begin{align*}
f_{k+1}(\bm \theta ) &= \frac{\mathbb P(C_k=\mathrm{ACK}|\mathcal B_{\mathrm t,k},\mathcal B_{\mathrm r,k},\xi_k=1,\bm  \Theta = \bm \theta )f_k(\bm \theta  )}{\int_{-\pi}^{\pi}\mathbb P(C_k=\mathrm{ACK}|\mathcal B_{\mathrm t,k},\mathcal B_{\mathrm r,k},\xi_k=1,\bm  \Theta = \tilde{\bm \theta}  ) f_k( \tilde{\bm \theta}  )  d\tilde{\bm \theta} }\\
&= \frac{\chi_{\mathcal{B}_k }(\bm \theta)}{\int_{\mathcal{B}_k }f_k(\tilde{\bm \theta})\mathrm d\tilde{\bm \theta} }f_k(\bm \theta),
\numberthis
\end{align*}
where $\mathcal B_k \triangleq \mathcal B_{\mathrm t,k} \times \mathcal B_{\mathrm r,k}$.
Similarly, for $k \in \mathcal I_s$ and $C_k=\mathrm{NACK}$, (\ref{eq:markov_property}) can be used to get
\begin{align*}
f_{k+1}(\bm \theta) = \frac{1-\chi_{\mathcal{B}_{k}}(\bm \theta )}{1-\int_{\mathcal{B}_k} f_k(\tilde{\bm \theta })\mathrm d\tilde{\bm \theta }}f_k(\bm \theta).\numberthis
 \end{align*}
For $k\in \mathcal I_d$, $\mathbb P(C_k=\mathrm{NULL}|\mathcal B_{\mathrm t,k},\mathcal B_{\mathrm r,k},\xi_k=0, \bm \Theta = \bm \theta )=1$. Therefore, we use (\ref{eq:markov_property}) to get
\begin{align}
f_{k+1}(\bm \theta ) =f_k(\bm \theta).
\end{align}
Thus we have proved the Lemma.
\end{proof}

We prove the lemma by induction. The hypothesis holds trivially for $k=0$. Let us assume that it holds in slot $k\geq 0$, we show that it holds in slot $k+1$ as well. First, let us consider the case when $k \in \mathcal I_{\mathrm s}$ and $C_k=\mathrm{ACK}$. By using (\ref{eq:belief_update})
along with the induction hypothesis,
we get
\begin{align*}
\label{eq:bs_ba_belief}
f_{k+1}(\bm \theta) 
&= \frac{f_0(\bm \theta ) }{\int_{-\pi}^{\pi}\chi_{{\mathcal U}_k \cap \mathcal{B}_k } (\tilde{\bm \theta} ) f_0(\tilde{\bm \theta} )d\tilde{\bm \theta} }\chi_{{\mathcal U}_k \cap \mathcal{B}_k } (\bm \theta).
\numberthis
\end{align*}
By substituting  ${\mathcal U}_{k+1}\equiv {\mathcal U}_{k} \cap \mathcal B_{k}$, we get (\ref{eq:belief_construct1}).
\par Next, we focus on the case when $k\in \mathcal I_{ \mathrm s}$ and $C_k=\mathrm{NACK}$. In this case,  (\ref{eq:belief_update}) yields
\begin{align*}
f_{k+1}(\theta_{\mathrm t},\theta_{\mathrm r})
= \frac{f_0(\bm \theta ) }{\int_{-\pi}^{\pi}\chi_{{\mathcal U}_k \setminus \mathcal{B}_k } (\tilde{\bm \theta} ) f_0(\tilde{\bm \theta} )d\tilde{\bm \theta} }\chi_{{\mathcal U}_k \setminus \mathcal{B}_k } (\bm \theta),
\numberthis
\end{align*}
where we used the fact that $\chi_{[-\pi,\pi]^2 \setminus{ \mathcal A}}(x) \equiv 1-\chi_{\mathcal A}(x)$. By observing that ${\mathcal U}_{k+1}\equiv {\mathcal U}_{k} \setminus{\mathcal B_{k}}$, we get the expression for $f_{k+1}(\bm \theta )$, as given in (\ref{eq:belief_construct1}).

\par Finally, for $k \in \mathcal I_d$, (\ref{eq:belief_update}) yields $f_{k+1}(\bm \theta) =  f_k(\bm \theta)$. Therefore, from the induction hypothesis it follows that $f_{k+1}(\bm \theta)$ is given by  (\ref{eq:belief_construct1}) with ${\mathcal U}_{k+1}={\mathcal U}_{k}$. Hence, the lemma is proved.
\end{proof}

\section{Supplementary Lemma \ref{optbeam}}
\label{proofofoptbeam}
\begin{lemma}
\label{optbeam}
The optimal 2D beam satisfies
\begin{align}
\label{eq:BS_UE_beams_listen}
\begin{cases}
&\mathcal B_{k} \subset \mathcal U_{k}, \forall k \in \mathcal I_{\mathrm s} \\
&\mathcal B_{k} \subseteq \mathcal U_{k}, \forall k \in \mathcal I_{\mathrm d}.
\end{cases}
\end{align}
\end{lemma}
\begin{proof}
We prove this lemma by contradiction. First, we consider the beam-alignment action $\mathbf a_k = (1,\mathcal B_k, 0)$ such that $\mathcal B_k \setminus \mathcal U_k \neq \emptyset$, i.e.,  ${\mathcal B}_{k}$ has non-empty support outside of ${\mathcal U}_{k}$.
Let  $\tilde{\mathbf a}_k = (1, \tilde{\mathcal B_k}, 0)$ be new beam-alignment action such that
$\tilde{\mathcal B}_{k}={\mathcal U}_{k}\cap{\mathcal B}_{k}$, i.e.,
$\tilde{\mathcal B}_{k}$ is constructed by restricting ${\mathcal B}_{k}$ within the belief support
${\mathcal U}_{k}$. Using
\eqref{eq:cost_to_go_iter} , we get
\begin{align*}
\label{eq:st1}
\hat V_k(\mathbf a_k; \mathcal U_k, D_k) =& \phi_s \size{\mathcal B_k} + \mathbb P (C_k =\mathrm{ACK}|\mathcal U_k , \mathcal B_k)  \hat V_{k+1}^*(\mathcal U_k \cap \mathcal B_k, D_k) \\
&+  \mathbb P (C_k =\mathrm{NACK}|\mathcal U_k , \mathcal B_k) \hat V_{k+1}^*(\mathcal U_k \setminus \mathcal B_k, D_k). \numberthis
\end{align*}
Using the fact that $\tilde{\mathcal B}_{k}={\mathcal U}_{k}\cap{\mathcal B}_{k}$, hence ${\mathcal U}_{k}\setminus {\mathcal B}_{k} = {\mathcal U}_{k}\setminus \tilde{\mathcal B}_{k}$, it follows that  $\mathbb P (C_k =c |\mathcal U_k , \mathcal B_k) = \mathbb P (C_k =c |\mathcal U_k ,\tilde{\mathcal B}_k),\; \forall c\in \{ \mathrm{ACK}, \mathrm{NACK}\}  $. Therefore, we rewrite  \eqref{eq:st1} as
\begin{align*}
\label{eq:st1}
&\hat V_k(\mathbf a_k; \mathcal U_k, D_k)
= \phi_s \size{\tilde{\mathcal B}_k }
+\phi_s \size{\mathcal U_k\setminus{\mathcal B}_k }
\\&
  + \mathbb P (C_k =\mathrm{ACK}|\mathcal U_k , \tilde{\mathcal B}_k)  \hat V_{k+1}^*(\tilde{\mathcal B}_k , D_k) 
+  \mathbb P (C_k =\mathrm{NACK}|\mathcal U_k , \tilde{\mathcal B}_k) \hat V_{k+1}^*(\mathcal U_k \setminus \tilde{\mathcal B}_k, D_k) \\
&> 
\hat V_k(\tilde{\mathbf a}_k; \mathcal U_k, D_k)
, \numberthis
\end{align*}
where we have used $\size{\mathcal U_k\setminus{\mathcal B}_k } >0$. Thus $\mathbf a_k$ is suboptimal, implying that optimal beam-alignment beam satisfy ${\mathcal B}_k \subseteq \mathcal U_k$.
Now, let $\mathcal B_k=\mathcal U_k$, and consider a new action with beam $\tilde{\mathcal B}_k=\emptyset$.
 Using a similar approach, it can be shown that $\mathcal B_k=\mathcal U_k$
is suboptimal with respect to $\tilde{\mathcal B}_k$, hence we must have ${\mathcal B}_k \subset \mathcal U_k$. 
\par To prove the lemma for  $k \in \mathcal I_d$, consider the action $\mathbf a_k =(0,\mathcal B_k,R_k)$  such that $ \mathcal B_k \setminus \mathcal U_k \neq \emptyset$. Now consider a new action $\tilde{\mathbf a}_k =(0,\tilde{\mathcal B}_k,R_k)$  such that $ \tilde{\mathcal B}_k = {\mathcal B}_k \cap \mathcal U_k$. It can be observed that $\mathbb{P}(\bm \theta \in \mathcal B_{k} |\mathcal U_k) =  \mathbb{P}(\bm \theta \in \tilde{\mathcal B}_{k} |\mathcal U_k)$. The cost-to-function for the action $\mathbf a_k$ is given as
\begin{align*}
\hat V_k(\mathbf a_k; \mathcal U_k, D_k) &= \frac{\psi_d(R_k)}{\bar{F}_{\gamma}^{-1}\left( \frac{1-\epsilon}{\mathbb{P}(\bm \theta \in{\mathcal B}_{k} |\mathcal U_{k})}\big| \hat\gamma   \right)}  \size{\mathcal B_k} + \hat V_{k+1}^* (\mathcal U_k, D_k-TR_k)\\
&>\frac{\psi_d(R_k)}{\bar{F}_{\gamma}^{-1}\left( \frac{1-\epsilon}{\mathbb{P}(\bm \theta\in \tilde{\mathcal B}_{k} |\mathcal U_{k})}  \big| \hat\gamma \right)}  \size{\tilde{\mathcal B}_k} + \hat V_{k+1}^* (\mathcal U_k, D_k-TR_k)\\
&= \hat V_k(\tilde{\mathbf a}_k; \mathcal U_k, D_k),
\end{align*}
hence we must have $\mathcal B_k\subseteq\mathcal U_k$. The lemma is thus proved.
\end{proof}
\section{ Proof of Theorem \ref{th:opt_comm_beams}}
\label{app:proof_of_opt_comm}
\begin{proof}
For a data communication action $\mathbf a_k{\in}{\mathcal A}_{\mathrm{ext}}(\mathcal U,D)$,
the state transition is independent of $\mathcal B_{k}$ since 
 ${\mathcal U}_{k+1}{=}{\mathcal U}_{k}$ and $D_{k+1}{=}D_k{-}R_kT$. Hence, the optimal beam given $R_k$ is obtained by minimizing $c(\mathbf a_k;{\mathcal U}_k , D_k)$ in
 (\ref{eq:cost1}), yielding
\begin{align*}
 c(\mathbf a_k;{\mathcal U}_k , D_k)
&\stackrel{(a)}{=} {\psi_d(R_k)}\frac{\size{{\mathcal B}_k} }{{\bar F}_{\gamma}^{-1}\left(\frac{(1-\epsilon){\size{{\mathcal U}_{k}}}}{\size{{\mathcal B}_{k}}} \big|\hat\gamma\right) }
\\&
 \label{eq:comm_val_ftn_gamma_deter}
\stackrel{(b)}{\geq }  {\psi_d(R_k)}{(1-\epsilon)} \size{{\mathcal U}_k}\frac{1}{q^*{\bar F}_{\gamma}^{-1}(q^*|\hat\gamma)}, \numberthis
\end{align*}
where (a) follows from $\mathbb{P}(\bm \theta{\in}\mathcal B_k |{\mathcal U}_k, \mathbf a_k) {=}\frac{\size{\mathcal B_k} }{\size{{\mathcal U}_k }}$, with
 $q{\triangleq}(1{-}\epsilon) \frac{\size{{\mathcal U}_k}}{\size{\mathcal B_k}}{\leq}1$ to enforce the $\epsilon$-outage constraint;
 (b) follows by maximizing  $q {\bar F}_{\gamma}^{-1}(q|\hat\gamma)$
 over $q{\in}[1{-}\epsilon,1]$.  Equality holds in  (b)
 if 
 $\size{\mathcal B_k }{=}\vartheta\size{\mathcal U_k} $, with $\vartheta{=}(1{-}\epsilon)/q^*$ and $q^*$ as in the statement. The theorem is thus proved.
\end{proof}
\vspace{-5mm}
\section{Proof of Theorem \ref{thm:decoupled}}
\label{proof_of_thm_decoupled}
\begin{proof}
Note that, if this policy satisfies
 $\mathcal B_k{\equiv}\mathcal B_{\mathrm t, k}{\times}\mathcal B_{\mathrm r, k}{\subseteq} \mathcal U_k\equiv\supp(f_k)$,
 along with the appropriate fractional values $\size{\mathcal B_k}/\size{\mathcal U_k}$,
then it is optimal since it satisfies all the conditions of Theorems~\ref{th:opt_comm_beams}
and \ref{optimalpoliuniform}.
We  now verify these conditions.
Since  $\mathcal B_{\mathrm t, k}\subseteq\mathcal U_{\mathrm t, k}$ and
$\mathcal B_{\mathrm r, k}\subseteq\mathcal U_{\mathrm r, k}$,
 it is sufficient to prove that $\mathcal U_{\mathrm t, k}\times\mathcal U_{\mathrm r, k}\equiv\mathcal U_k,\forall k$.
Indeed, $\mathcal U_0 \equiv \mathcal U_{\mathrm t, 0} \times \mathcal U_{\mathrm r,0}$.  
By induction, assume  that $\mathcal U_{k}\equiv \mathcal U_{\mathrm t,k}\times\mathcal U_{\mathrm r, k}$. Then, for $\beta_k = 1$ (a similar result holds for $\beta_k=2$), using \eqref{eq:supp_update} we get
\begin{align}
\mathcal U_{k+1} = \begin{cases}
( \mathcal U_{\mathrm t, k} \cap  \mathcal B_{t, k} ) \times  \mathcal U_{r, k}, &\text{ if } C_{k}= \mathrm{ACK},\\
( \mathcal U_{\mathrm t, k} \setminus \mathcal B_{t, k} ) \times  \mathcal U_{r, k},  &\text{ if } C_{k}= \mathrm{NACK}.
\end{cases} 
\end{align}
By letting $\mathcal U_{\mathrm r, k}{\equiv}\mathcal U_{\mathrm r, k-1}$,  $\mathcal U_{\mathrm t, k}{\equiv}\mathcal U_{\mathrm t, k-1}{\cap}\mathcal B_{t, k-1}$ if $C_k{=} \mathrm{ACK}$ and  $\mathcal U_{\mathrm t, k}{\equiv}\mathcal U_{\mathrm t, k-1} {\setminus}\mathcal B_{t, k-1}$ if $C_k{=}\mathrm{NACK}$, we obtain $\mathcal U_k{\equiv}\mathcal U_{\mathrm t, k}{\times}\mathcal U_{\mathrm r, k}$. This policy is then optimal. 
Finally, \eqref{Powunif} is obtained by using the relation between power consumption and value function.
 Thus, we have proved the theorem.
\end{proof}
\section{Proof of Theorem \ref{th:perf_with_errors}}
\label{app:proof_of_th_perf_with_errors}
\begin{proof}
We prove it by induction using the DP updates.
Let  $\bar T_{k}(\mathcal U_{k},e_{k})$ be the \emph{throughput-to-go} function from
state $(\mathcal U_k,e_k)$ in slot $k\leq L^*$.
We prove by induction that
\begin{align} 
\nonumber
\bar T_{k}(\mathcal U_{k},e_{k})
=&
(1-e_{k})(1-\epsilon) R_{\min}
\\&\prod_{j=k}^{L^*-1}\left[\left(1-\rho_k\right)(1-p_{\mathrm{fa}})+\rho_k(1-p_{\mathrm{md}}) \right].
\end{align}
Then, \eqref{eq:thr_frac} follows from 
$\bar T_{\mathrm{err}}{=}\bar T_{0}(\mathcal U_{0},0)$. The induction hypothesis holds at $k{=}L^*$, since 
$\bar T_{L^*}(\mathcal U_{L^*},e_{L^*})=(1{-}e_{L^*})(1{-}\epsilon) R_{\min}$, see \eqref{eq:avg_th_1}. Now, assume it holds for some $k{+}1\leq L^*$. Using the transition probabilities from state 
$(\mathcal U_k,1)$ and the induction hypothesis, we obtain $\bar T_{k}(\mathcal U_{k},1)=0$.
Instead, from state
$(\mathcal U_k,0)$ we obtain
\begin{align*}
&\bar T_{k}(\mathcal U_{k},0)=
\rho_k(1-p_{\mathrm{md}})
\bar T_{k+1}(\mathcal B_{k},0)
\\&
\qquad\qquad \quad+
\left(1-\rho_k\right)(1-p_{\mathrm{fa}})
\bar T_{k+1}(\mathcal U_k\setminus\mathcal B_{k},0)
\\&
=
(1-\epsilon) R_{\min} \prod_{j=k}^{L^*-1}\Bigr[\left(1-\rho_k\right)(1-p_{\mathrm{fa}})
+\rho_k(1-p_{\mathrm{md}}) \Bigr],
\end{align*}
which readily follows by applying the induction hypothesis. The induction step is thus proved.

Let  $\bar E_{k}(\mathcal U_{k},e_{k})$ be the \emph{energy-to-go} from
state $(\mathcal U_k,e_k)$ in slot $k{\leq}L^*$.
We prove that
\begin{align}
\bar E_{k}(\mathcal U_{k},e_{k})
=\left[v_k^{(L^*)}+h_k+u_k(1-e_k)\right]\size{\mathcal U_{k}}.
\end{align}
Then, \eqref{eq:thr_frac} follows from 
$\bar P_{\mathrm{err}}=\frac{1}{T_{\mathrm{fr}}}\bar E_{0}(\mathcal U_{0},0)$,
and by noticing that 
$v_0^{(L^*)}/T_{\mathrm{fr}}$ is the power consumption in the error-free case, given in Theorem \ref{thm:decoupled}.
 The induction hypothesis holds at $k{=}L^*$, since 
$\bar E_{L^*}(\mathcal U_{L^*},e_{L^*}){=}
(N-L^*)\phi_d\left(\frac{NR_{\min}}{N-L^*},\epsilon\right)\size{\mathcal U_{L^*}}=v_{L^*}^{(L^*)}+h_{L^*}+u_{L^*}(1-e_{L^*})$,
with $v_{L^*}^{(L^*)}$ given by \eqref{dfh},  $h_{L^*}=u_{L^*}=0$,
 see \eqref{eq:avg_th_1}.
  Now, assume it holds for some $k+1\leq L^*$. Using the transition probabilities from state 
$(\mathcal U_k,e_k)$, the induction hypothesis, and
the fact that $\size{\mathcal B_{k}}=\rho_k\size{\mathcal U_{k}}$ and
$\size{\mathcal U_k\setminus\mathcal B_{k}}=(1-\rho_k)\size{\mathcal U_{k}}$,
 we obtain
\begin{align*}
\bar E_{k}(\mathcal U_{k},1)&=
\phi_s\rho_k\size{\mathcal U_{k}}+p_{\mathrm{fa}}\bar E_{k+1}(\mathcal B_{k},1)
+(1-p_{\mathrm{fa}})\bar E_{k+1}(\mathcal U_{k}\setminus\mathcal B_k,1)
\\
&{=}
\Bigr\{
\phi_s\rho_k
{+}\left(v_{k+1}^{(L^*)}{+}h_{k+1}\right)\left[p_{\mathrm{fa}}\rho_k{+}(1-p_{\mathrm{fa}})\left(1-\rho_k\right)\right]
\Bigr\}\size{\mathcal U_{k}};
\end{align*}

\begin{align*}
\bar E_{k}(\mathcal U_{k},0)= &
\phi_s\rho_k\size{\mathcal U_{k}}
+\rho_k(1-p_{\mathrm{md}})\bar E_{k+1}(\mathcal B_{k},0)
+\left(1-\rho_k\right)p_{\mathrm{fa}}\bar E_{k+1}(\mathcal B_{k},1)
\\&
+\left(1-\rho_k\right)(1-p_{\mathrm{fa}})\bar E_{k+1}(\mathcal U_k\setminus\mathcal B_{k},0)
+\rho_kp_{\mathrm{md}}\bar E_{k+1}(\mathcal U_k\setminus\mathcal B_{k},1)
\\
=&
\Biggr\{
\phi_s\rho_k
+\left(v_{k+1}^{(L^*)}+h_{k+1}+u_{k+1}\right)\Bigr[\rho_k^2(1-p_{\mathrm{md}})
+\left(1-\rho_k\right)^2(1-p_{\mathrm{fa}})\Bigr]
\\&
\;\;{+}\left(v_{k+1}^{(L^*)}{+}h_{k+1}\right)\rho_k\left(1-\rho_k\right)(p_{\mathrm{fa}}{+}p_{\mathrm{md}})
\Biggr\}
\size{\mathcal U_{k}}.
\end{align*}
The induction step $\bar E_{k}(\mathcal U_{k},e_k)=(v_k^{(L^*)}+h_k+u_k(1-e_k))\size{\mathcal U_k}$ can be finally proved by 
 expressing
$v_{k}^{(L^*)}=g_k(\rho_k)$ and $\rho_k=\frac{1}{2}-\frac{\phi_s}{4v_{k+1}^{(L^*)}}$ using \eqref{vkupdate}, and using \eqref{eq:h_k}-\eqref{eq:u_k}.
\end{proof}

\bibliographystyle{IEEEtran}
\bibliography{IEEEabrv,biblio}

\end{document}